\documentclass[10pt]{article}
\usepackage[margin=2cm]{geometry}

\usepackage{bbm}
\usepackage{amsmath}
\usepackage{amsfonts}
\usepackage{amssymb}
\usepackage{graphicx}
\usepackage{color}
\usepackage{geometry}%
\usepackage{accents}
\usepackage{mathrsfs}
\usepackage{exscale}
\usepackage{subfigure}
\usepackage{float}
\usepackage[normalem]{ulem}

\makeatletter
\def\wideubar{\underaccent{{\cc@style\underline{\mskip10mu}}}}
\def\Wideubar{\underaccent{{\cc@style\underline{\mskip8mu}}}}
\makeatother

\makeatletter
\def\widebar{\accentset{{\cc@style\underline{\mskip10mu}}}}
\def\Widebar{\accentset{{\cc@style\underline{\mskip8mu}}}}
\makeatother

\allowdisplaybreaks

\newtheorem{theorem}{Theorem}[section]

\newtheorem{corollary}{Corollary}[section]

\newtheorem{proposition}{Proposition}[section]

\definecolor{DarkGreen}{rgb}{0.2,0.6,0.2}

\newenvironment{proof}[1][Proof]{\noindent\textbf{#1.} }{\ \rule{0.5em}{0.5em}}

\numberwithin{equation}{section}

\numberwithin{equation}{section}

\newcommand{\be}{\begin{equation}}
\newcommand{\ee}{\end{equation}}

\newcommand{\bq}{\begin{eqnarray}}
\newcommand{\eq}{\end{eqnarray}}
\newcommand{\essinf}{\displaystyle\mathop{\mathrm{ess\,inf}}\displaylimits}

\begin{document}
\title{An FBSDE approach to market impact games\\ with stochastic parameters}

\author{Samuel Drapeau\thanks{SAIF/CAFR/CMAR and School of Mathematical Sciences, Shanghai Jiao Tong University, China; Email: sdrapeau@saif.sjtu.edu.cn} \and Peng Luo\thanks{Department of Statistics and Actuarial Sciences, University of Waterloo, Canada; Email: peng.luo@uwaterloo.ca} \and Alexander Schied\thanks{Department of Statistics and Actuarial Sciences, University of Waterloo, Canada; Email: aschied@uwaterloo.ca} \and Dewen Xiong\thanks{School of Mathematical Sciences, Shanghai Jiao Tong University, China; Email: xiongdewen@sjtu.edu.cn}}

\maketitle
\begin{abstract}
We analyze a market impact game between  $n$ risk averse agents who compete for liquidity in a market impact model with permanent price impact and additional slippage. Most market parameters, including volatility and drift, are allowed to vary stochastically. Our first main result characterizes the Nash equilibrium in terms of a fully coupled system of forward-backward stochastic differential equations (FBSDEs). Our second main result provides conditions under which this system of FBSDEs has indeed a unique solution, which in turn yields the unique Nash equilibrium. We furthermore obtain closed-form solutions in special situations and analyze them numerically. \end{abstract}

\section{Introduction}

Market impact games analyze situations in which several agents compete for liquidity in a market impact model or try to exploit the price impact generated by competitors. In this paper, we follow Carlin et al. \cite{Carlinetal}, Sch\"{o}neborn and Schied \cite{SchoenebornSchied}, Carmona and Yang \cite{CarmonaYang}, Schied and Zhang \cite{SchiedZhangCARA}, Casgrain and Jaimungal \cite{CasgrainJaimungal}, and others by analyzing a market impact game in the context of the Almgren--Chriss market impact model. In \cite{Carlinetal,SchoenebornSchied}, all agents are risk-neutral and market parameters are constant, which leads to deterministic Nash equilibria. Deterministic open-loop equilibrium strategies are also obtained in \cite{SchiedZhangCARA}, where agents maximize mean variance functionals or CARA utility.  In \cite{CarmonaYang} closed-loop equilibria are studied numerically in a similar setup, and it is found by means of simulations that then equilibrium strategies may no longer be deterministic. The approach in  \cite{CasgrainJaimungal} is the closest to ours. There, the authors analyze the infinite-agent, mean-field limit of a market impact game for heterogeneous, risk-averse agents in a model with constant coefficients and partial information, and they characterize the mean-field game through a forward-backward stochastic differential equation (FBSDE). In addition, there are several papers that study market impact games in other price impact models, including models with linear transient price impact; see, e.g., \cite{Moallemietal,SchiedZhangHotPotato,SchiedStrehleZhang, LuoSchied}.

Our contribution to this literature is twofold.
First, on the mathematical side, we completely solve the problem of determining an open-loop Nash equilibrium with stochastic model parameters and risk aversion for arbitrary numbers of agents. Our solution relies on a  characterization of  the equilibrium strategies in terms of a fully coupled systems of forward-backward stochastic differential equations (FBSDEs). This characterization is given in Theorem \ref{thm 4.1}. In the subsequent Theorem \ref{thm 4.2}, we give sufficient conditions that guarantee the existence of a unique solution. The main restriction is a lower bound on the volatility. Then we analyze the case of constant coefficients and the case in which all agents share the same parameters but have different initial inventories. Numerical simulations are provided for the case of constant coefficients which work for many agents.

Our second contribution consists in a modification of the traditional setup of the interaction term in a market impact game with Almgren--Chriss-style price impact. The Almgren--Chriss model has two price impact components, one permanent and one temporary. It is clear that permanent price impact must affect the execution prices of all agents equally, and in \cite{Carlinetal,SchoenebornSchied,CarmonaYang,SchiedZhangCARA,CasgrainJaimungal} the same is assumed of the transient price impact. This assumption can sometimes lead to counterintuitive results. For instance, if the temporary price impact is large in comparison with the permanent price impact, then, in the presence of a large seller, it can be beneficial to build up a long position in the stock, because a cessation of the trading activities of the large seller will lead to an immediate upwards jump of the expected price \cite{SchoenebornSchied}. In the price impact literature, it is however not consensus that  \lq\lq temporary price impact" is of the same nature as permanent price impact. For instance,  Almgren et al.~\cite{AlmgrenHauptmanLi} write about temporary  impact:
\begin{quote}
This expression is a continuous-time approximation to a discrete process. A more accurate description would be to imagine that time is broken into intervals such as, say, one hour or one half-hour. Within each interval, the average price we realise on our trades during that interval will be slightly less favorable than the average price that an unbiased observer would measure during that time interval. The unbiased price is affected on previous trades that we have executed before this interval (as well as volatility), but not on their timing. The additional concession during this time interval is strongly dependent on the number of shares that we execute in this interval.
\end{quote}
Likewise, Gatheral \cite[p. 751]{Gatheral} writes:
\begin{quote}
The second component of the cost of trading corresponds to market frictions such as effective bid-ask spread that affect only our execution price: We refer to this component of trading cost as \emph{slippage} (\emph{temporary impact} in the terminology of Huberman and Stanzl).
\end{quote}
Based on these interpretations of \lq\lq temporary price impact" as slippage, it appears to be more natural that only the trades of the executing agent and not the trades of the other market participants are affected by the resulting cost. In our paper, we therefore keep a term for \lq\lq temporary price impact", but it only affects the execution costs of the corresponding agent and not of the other agents.

The paper is organized as follows. In Section 2, we set up our model on portfolio liquidation in the Almgren-Chriss framework. Single agent optimization is studied in Section 3, where the corresponding existence, uniqueness and characterization results for the optimal liquidation strategy are stated. Section 4 is dedicated to present the characterization result for Nash equilibrium and investigates the solvability of the characterizing FBSDE. Some explicit solutions for Nash equilibria are analyzed in Section 5.

\section{Preliminaries and problem formulation}

\subsection{Frequently used notation}

Let $W=(W_t)_{t\geq 0}$ be a $d$-dimensional Brownian motion on a probability space $(\Omega, {\cal F}, P)$ and denote by $(\mathcal{F}_t)_{t\geq 0}$  the complete filtration generated by $W$. Throughout, we fix a finite time horizon $T>0$. We endow $\Omega \times [0,T]$ with the predictable $\sigma$-algebra $\mathcal{P}$ and $\mathbb{R}^n$ with its Borel $\sigma$-algebra $\mathcal{B}(\mathbb{R}^n)$. Equalities and inequalities between random variables and processes are understood in the $P$-a.s.~and $P\otimes dt$-a.e.~sense, respectively. The Euclidean norm is denoted by $|\cdot|$.  For $m\in[1,\infty]$ and $k\in\mathbb N$, we denote by  $\|\cdot\|_m$ denotes the $L^m$-norm, by
 $\mathcal{S}^m(\mathbb{R}^k)$ the set of $k$-dimensional continuous adapted processes $Y$ on $[0,T]$ such that
\begin{equation*}
\|Y\|_{\mathcal{S}^m(\mathbb{R}^k)}:=\left\|\sup_{0\leq t\leq T} |Y_t|^m\right\|_m< \infty,\end{equation*}
and by  $\mathcal{H}^m(\mathbb{R}^{k})$ the set of predictable $\mathbb{R}^{k}$-valued processes $Z$ such that
    \begin{equation*}
    \|Z\|_{\mathcal{H}^m(\mathbb{R}^{k})}=\left\|\left(\int_0^T|Z_s|^2ds\right)^{\frac{1}{2}}\right\|_m<\infty.
    \end{equation*}
The space $\text{\rm BMO}(\mathbb{R}^{k})$ consists of all predictable $\mathbb{R}^{k}$-valued processes $Z$ such that
    \begin{equation*}
    \|Z\|_{\text{\rm BMO}(\mathbb{R}^{k})}=\sup_{\tau\in\mathcal{T}}\Bigg\|E\left[\left(\int_{\tau}^T|Z_s|^2ds\right)^{\frac{1}{2}}\Bigg|\mathcal{F}_{\tau}\right]\Bigg\|_{\infty}<\infty
    \end{equation*}
where $\mathcal{T}$ is the set of all stopping times with values in $[0,T]$.

\subsection{Model  setup}\label{model setup section}

We consider $n$ financial agents who are active in a  financial market of Almgren--Chriss-type and whose trading strategies interact via permanent price impact. More precisely, we adapt the continuous-time setting of \cite{Almgren}, where each agent $i$ has initial inventory $Q_0^i$ at time $t=0$ and subsequently uses a trading strategy whose trading rate is given by a process $q^i\in\mathcal{H}^2(\mathbb{R})$. That is, at time $t\in[0,T]$, the inventory of agent $i$ is given by
\begin{equation*}
Q^{q^i}_t=Q^i_0+\int_0^tq^i_sds.
\end{equation*}
This trading strategy impacts the price of the risky asset by means of permanent price impact. It is usually assumed that this permanent price impact is linear in the traded inventory (see, e.g., the discussion in Section 3 of \cite{GatheralSchiedSurvey}). Thus, we assume that the price at which shares of the risky assets can be traded at time $t$ is given by
\begin{equation}\label{asset price process eq}
S^q_t=S_0+\int_0^t\mu_s\,ds+a\sum_{i=1}^n\int_0^tq^i_s\,ds+\int_0^t\sigma_s\,dW_s,
\end{equation}
where $\mu\in\mathcal{S}^{\infty}(\mathbb{R})$ is a generic drift, $\sigma\in\mathcal{S}^{\infty}(\mathbb{R}^{d})$ is a volatility process, and, for a fixed price impact parameter $a>0$,  the term $a\sum_{i=1}^n\int_0^tq^i_s\,ds$ describes the cumulative price impact generated by the strategies of all agents.

At time $t$, the $i^{\text{th}}$ agent sells $-q^i_t\,dt$ shares at price $S^q_t$. The implementation shortfall, i.e., the difference between book value and liquidation proceeds, is therefore given by $Q^i_0S_0-Q^{q^i}_TS^q_T+\int_0^Tq^i_tS^q_t\,dt$.
In addition, the trading strategy $q^i$ generates \lq\lq slippage", including transaction costs, instantaneous price impact effects etc., modeled by the cost functional $b\int_0^T(q^i_t)^2\,dt$; see, e.g., \cite{Almgren} and the discussion in the introduction. Moreover, any inventory held at time $t>0$ gives rise to financial risk. We assume that this risk is measured by the expectation of the term
\begin{equation}\label{risk term eq}
\alpha_i \left(Q^{q^i}_T\right)^2+\int_0^T\lambda_i\sigma^2_t\left(Q^{q^i}_t\right)^2dt
\end{equation}
where $\alpha_i$ and $\lambda_i$ are nonnegative  constants. The first term in \eqref{risk term eq}
 is clearly a penalty term penalizing any inventory that is still present at time $T$. As shown by Sch\"oneborn \cite{SchoenebornTimeInconsistent,Schoeneborn}, the expectation of the integral term in \eqref{risk term eq} can be regarded as  a continuously re-optimized variance functional with infinitesimal time horizon; see also  \cite{AlmgrenSIFIN,Forsythetal,Tseetal} for related motivations of this risk term.
It follows that the objective of agent $i$ is to minimize the expectation of following cost functional over strategies $q^i\in\mathcal{H}^2(\mathbb{R})$,
  \begin{equation}\label{cost-risk functional eq}
         C^i_T(Q^i_0,q^i,q^{-i})=Q^i_0S^q_0 +\int_0^Tq^i_t\left(S^q_t+b q^i_t\right)dt-Q^{q^i}_TS^q_T+\alpha_i \left(Q^{q^i}_T\right)^2+\int_0^T\lambda_i\sigma^2_t\left(Q^{q^i}_t\right)^2dt;
         \end{equation}
         here,  $q^{-i}:=(q^1,\ldots,q^{i-1},q^{i+1},\ldots,q^{n})$ denotes the collection of the strategies of all other agents.

   Our goal in this paper is to discuss the existence, uniqueness and structure of Nash equilibria for the cost criterion described above. As usual, a collection   $q^*=(q^{1*},\ldots,q^{n*})\in\mathcal{H}^{2}(\mathbb{R}^n)$ of strategies will be called a \emph{Nash equilibrium} if,  for $i=1,\ldots,n$,
         \begin{equation*}
      \min_{q^i\in\mathcal{H}^2(\mathbb{R})}E\left[C^i_T(Q^i_0,q^i,q^{-i*})\right]=E\left[C^i_T(Q^i_0,q^{i*},q^{-i*})\right].
      \end{equation*}

\section{Single-agent optimization}

In preparation for the discussion of Nash equilibria defined at the end of  Section \ref{model setup section}, we analyze first the optimization problem for a fixed agent $i$ when the strategies of all other agents are fixed. A variety of methods has been used to solve similar and related problems; see, e.g., \cite{AlmgrenSIFIN,Forsythetal,Tseetal,SchiedFuel,GraeweHorstQui,AnkirchnerJeanblancKruse}. Here, our goal is to represent solutions in terms of a BSDE in Theorem \ref{Thm1}.

First, plugging formula \eqref{asset price process eq}
 for $S^q$ into our expression \eqref{cost-risk functional eq}
 of the cost-risk functional $C^i_T(Q^i_0,q^i,q^{-i})$ and integrating by parts, we obtain the alternative expression
   \begin{align*}
   C^i_T(Q^i_0,q^i,q^{-i})=&\frac{a}{2}\left(Q^i_0\right)^2-\int_0^TQ^{q^i}_t\left(\mu_t +a\sum_{j\neq i}q^j_t\right)dt -\int_0^TQ^{q^i}_t\sigma_tdW_t\\
   &+\int_0^T b \left(q^i_t\right)^2dt+\left(\alpha_i-\frac{a}{2}\right) \left(Q^{q^i}_T\right)^2+\int_0^T\lambda_i\sigma^2_t\left(Q^{q^i}_t\right)^2dt.
   \end{align*}
Since, by assumption,  $\sigma\in\mathcal{S}^{\infty}(\mathbb{R}^{d})$ and $q^i\in\mathcal{H}^2(\mathbb{R})$, the stochastic integral $\int_0^TQ^{q^i}_t\sigma_tdW_t$ is a true martingale, and so taking expectations yields
   \begin{align*}
   E\left[C^i_T(Q^i_0,q^i,q^{-i})\right]
   &=\frac{a}{2}\left(Q^i_0\right)^2-E\left[\int_0^TQ^{q^i}_t\left(\mu_t +a\sum_{j=1}^nq^j_t\right)dt\right] +E\left[\int_0^T b \left(q^i_t\right)^2d t\right]\\
   &\quad +\left(\alpha_i-\frac{a}{2}\right)E\left[\left(Q^{q^i}_T\right)^2\right] +E\left[\int_0^T\lambda_i\sigma^2_t\left(Q^{q^i}_t\right)^2dt\right].
   \end{align*}
In the following, we will denote $\beta_i=\alpha_i-\frac{a}{2}$. Fixing $0\leq t\leq T$, let
\begin{equation*}
Q^{\bar{q}^i}_{t,s}:=Q^{q_i}_t+\int_t^s\bar{q}^i_udu,\quad \text{for }t\leq s\leq T,
\end{equation*}
and $C^i_{t,T}(Q^{q^i}_t,\bar{q}^i,q^{-i})$ be the total cost on $[t,T]$ if, at time $t$,  agent $i$ starts using the strategy  $\bar{q}^i$ with the inventory $Q^{q^i}_t$, i.e.,
   \begin{align*}
   C^i_{t,T}(Q^{q^i}_t,\bar{q}^i,q^{-i})
   =&\frac{a}{2}\left(Q^{q^i}_t\right)^2-\int_t^TQ^{\bar{q}^i}_{t,u}\left(\mu_u +a\sum_{j\neq i}q^j_u\right) du +\int_t^T\lambda_i\sigma^2_u\left(Q^{\bar{q}^i}_{t,u}\right)^2du\\
     &+\int_t^T b \left(\bar{q}^i_u\right)^2du-\int_t^TQ^{\bar{q}^i}_{t,u}\sigma_udW_u+\beta_i\left(Q^{\bar{q}^i}_{t,T}\right)^2.
   \end{align*}
Let
   \begin{align*}
   \Phi^i_t\left(Q^{q^i}_t\right):&=\essinf_{\bar{q}^i\in\mathcal{H}^2(\mathbb{R})} E\left[C^i_{t,T}\left(Q^{q^i}_t,\bar{q}^i,q^{-i}\right)\Big|\mathcal{F}_t\right]\\
   &=\frac{a}{2}\left(Q^{q^i}_t\right)^2+\essinf_{\bar{q}^i\in\mathcal{H}^2\left(\mathbb{R}\right)}E\left[\int_t^T\left(-Q^{\bar{q}^i}_{t,u}\left(\mu_u +a\sum_{j\neq i}q^j_u-\lambda_i\sigma^2_uQ^{\bar{q}^i}_{t,u}\right)+b \left(\bar{q}^i_u\right)^2\right)du+\beta_i\left(Q^{\bar{q}^i}_{t,T}\right)^2\Bigg|\mathcal{F}_t\right].
   \end{align*}

Our next goal is to obtain a representation of
   \begin{equation*}
    \hat{\Phi}^i_t\left(Q^{q^i}_t\right):=\Phi^i_t\left(Q^{q^i}_t\right)-\frac{a}{2}\left(Q^{q^i}_t\right)^2.
   \end{equation*}
   in terms of component $(A^i,B^i,C^i)$ of a solution of a three-dimensional BSDE, which will be discussed in the following proposition.

\begin{proposition}\label{Thm1}
Suppose that $\beta_i\geq 0$ and $q^j\in\mathcal{H}^2(\mathbb{R}),~j\neq i$, then the following BSDE
\begin{equation*}
\begin{cases}
 &A^i_t=\beta_i-\int_t^T\left(\frac{1}{b} \left(A^i_s\right)^2-\lambda_i\sigma^2_s\right)ds-\int_t^TZ^{A^i}_sdW_s,\\
 &B^i_t=0-\int_t^T\left(\frac{1}{b}A^i_sB^i_s +\mu_s+a\sum_{j\neq i}q^j_s\right)ds-\int_t^TZ^{B^i}_sdW_s\\
 \end{cases}
\end{equation*}
admits a unique solution $(A^i,B^i,Z^{A^i},Z^{B^i})\in\mathcal{S}^{\infty}(\mathbb{R})\times\mathcal{S}^2(\mathbb{R})\times \text{\rm BMO}(\mathbb{R}^d)\times\mathcal{H}^2(\mathbb{R}^d)$. Moreover, the solution of the BSDE
$$ dC^i_t=\frac{1}{4b}\left(B^i_t\right)^2dt+Z^{C^i}_tdW_t,\quad C^i_T=0,
$$ is well defined and given by
\begin{equation*}
C^i_t=0-\int_t^T \frac{1}{4b}\left(B^i_s\right)^2ds-\int_t^TZ^{C^i}_sdW_s.
\end{equation*}
\end{proposition}

\begin{proof}
Denoting $M=\beta_i+\lambda_i\|\sigma\|_{\infty}^2T$, it follows from Pardoux and Peng \cite{PP} that BSDE
\begin{equation*}
 A^i_t=\beta_i-\int_t^T\left(\frac{1}{b}  \left( \left(-M\right)\vee A^i_s\wedge M\right)^2-\lambda_i\sigma^2_s\right)ds-\int_t^TZ^{A^i}_sdW_s
\end{equation*}
admits a unique solution $(A^i,Z^{A^i})\in\mathcal{S}^2(\mathbb{R})\times\mathcal{H}^2(\mathbb{R}^d)$. Moreover, we have the following estimate for $A^i$,
\begin{align*}
A^i_t&\leq E\left[\beta_i-\int_t^T\left(\frac{1}{b}  \left( \left(-M\right)\vee A^i_s\wedge M\right)^2-\lambda_i\sigma^2_s\right)du|\mathcal{F}_t\right]\leq \beta_i+\lambda_i\|\sigma\|_{\infty}^2(T-t).
\end{align*}
Meanwhile by denoting $\xi_t=\frac{\left(-M\right)\vee A^i_t\wedge M}{b}$, it holds that
\begin{align*}
e^{-\int_0^t\xi_sds}A^i_t&=e^{-\int_0^T\xi_sds}\beta_i+\int_t^Te^{-\int_0^s\xi_udu}\left(\lambda_i\sigma^2_s+\xi_sA^i-b\xi^2_s\right)ds-\int_t^Te^{-\int_0^s\xi_udu}Z^{A^i}_sdW_s\\
&\geq e^{-\int_0^T\xi_sds}\beta_i-\int_t^Te^{-\int_0^s\xi_udu}Z^{A^i}_sdW_s.
\end{align*}
Therefore, we have
\begin{align*}
A^i_t&\geq E\left[e^{-\int_t^T\xi_sds}\beta_i\Big|\mathcal{F}_t\right]\geq \beta_ie^{-\frac{M(T-t)}{b}}.
\end{align*}
Hence, $(A^i,Z^{A^i})\in\mathcal{S}^{\infty}(\mathbb{R})\times\mathcal{H}^2(\mathbb{R}^d)$ and satisfies
\begin{equation*}
A^i_t=\beta_i-\int_t^T\left(\frac{1}{b} \left(A^i_s\right)^2-\lambda_i\sigma^2_s\right)ds-\int_t^TZ^{A^i}_sdW_s.
\end{equation*}
It is easy to check that $Z^{A^i}\in \text{\rm BMO}(\mathbb{R}^d)$. On the other hand, if
\begin{equation*}
A^i_t=\beta_i-\int_t^T\left(\frac{1}{b} \left(A^i_s\right)^2-\lambda_i\sigma^2_s\right)du-\int_t^TZ^{A^i}_sdW_s
\end{equation*}
admits a solution $(A^i,Z^{A^i})\in\mathcal{S}^{\infty}(\mathbb{R})\times \text{\rm BMO}(\mathbb{R}^d)$, we have
\begin{align*}
A^i_t&\leq E\left[\beta_i-\int_t^T\left(\frac{1}{b} \left( A^i_s\right)^2-\lambda_i\sigma^2_s\right)ds\Bigg|\mathcal{F}_t\right]\leq \beta_i+\lambda_i\|\sigma\|_{\infty}^2(T-t)
\end{align*}
and
\begin{align*}
e^{-\int_0^t\frac{A^i_s}{b}ds}A^i_t&=e^{-\int_0^T\frac{A^i_s}{b}ds}\left(\beta_i\right)+\int_t^Te^{-\int_0^s\frac{A^i_u}{b}du}\lambda_i\sigma^2_sds-\int_t^Te^{-\int_0^s\frac{A^i_u}{b}du}Z^{A^i}_sdW_s\\
&\geq e^{-\int_0^T\frac{A^i_s}{b}ds}\left(\beta_i\right)-\int_t^Te^{-\int_0^s\frac{A^i_u}{b}du}Z^{A^i}_sdW_s.
\end{align*}
Therefore, we have
\begin{align*}
A^i_t&\geq E\left[e^{-\int_t^T\frac{A^i_s}{b}ds}\beta_i|\mathcal{F}_t\right]\geq \beta_ie^{-\frac{M(T-t)}{b}}.
\end{align*}
Hence, $(A^i,Z^{A^i})$ satisfies
\begin{equation*}
 A^i_t=\beta_i-\int_t^T\left(\frac{1}{b}  \left( \left(-M\right)\vee A^i_s\wedge M\right)^2-\lambda_i\sigma^2_s\right)ds-\int_t^TZ^{A^i}_sdW_s.
\end{equation*}
Again, it follows from Pardoux and Peng \cite{PP} that
\begin{equation*}
 B^i_t=0-\int_t^T\left(\frac{1}{b}A^i_sB^i_s +\mu_s+a\sum_{j\neq i}q^j_s\right)du-\int_t^TZ^{B^i}_sdW_s
\end{equation*}
admits a unique solution $(B^i,Z^{B^i})\in\mathcal{S}^2(\mathbb{R})\times\mathcal{H}^2(\mathbb{R}^d)$. The rest is clear.
\end{proof}
\begin{theorem}\label{Thmi}
Suppose that $\beta_i\geq 0$ and $q^j\in\mathcal{H}^2(\mathbb{R}),~j\neq i$, then $\hat{\Phi}^i_t\left(Q^{q^i}_t\right)$ is given by
   \begin{equation*}
   \hat{\Phi}^i_t\left(Q^{q^i}_t\right)= A^i_t\left(Q^{q^i}_t\right)^2
   +B^i_tQ^{q^i}_t+C^i_t
   \end{equation*}
where $A^i,B^i,C^i$ are given as in Proposition \ref{Thm1}. The unique optimal strategy  for the agent $i$ is given in feedback form by
\begin{equation*}
q^{i*}_t=-\frac{1}{b}\left(A^i_t Q^{q^{i*}}_t
   +\frac{1}{2}B^i_t\right).
\end{equation*}
\end{theorem}
\begin{proof}
By denoting
\begin{equation*}
\begin{cases}
  V^{q^i}_t&= A^i_t\left(Q^{q^i}_t\right)^2
   +B^i_tQ^{q^i}_t+C^i_t,\\
  g^{A^i}_t&=\frac{1}{b}\left(A^i_t\right)^2-\lambda_i\sigma^2_t,\\
  g^{B^i}_t&=\frac{1}{b}A^i_tB^i_t +\mu_t +a\sum_{j\neq i}q^j_t,\\
  g^{C^i}_t&=\frac{1}{4b}\left(B^i_t\right)^2,
\end{cases}
\end{equation*}
and applying It\^{o}'s formula, we have
   \begin{align*}
   dV^{q^i}_t= &2A^i_tQ^{q^i}_tq^i_tdt +\left(Q^{q^i}_t\right)^2dA^i_t+B^i_tq^i_tdt+Q^{q^i}_tdB^i_t+g^{C^i}_tdt+Z^{C^i}_tdW_t\\
   = &2A^i_tQ^{q^i}_tq^i_tdt +\left(Q^{q^i}_t\right)^2g^{A^i}_tdt+\left(Q^{q^i}_t\right)^2Z^{A^i}_tdW_t+g^{C^i}_tdt+Z^{C^i}_tdW_t+B^i_tq^i_tdt+Q^{q^i}_tg^{B^i}_tdt+Q^{q^i}_tZ^{B^i}_tdW_t\\
   = &\left(2A^i_tQ^{q^i}_tq^i_t +\left(Q^{q^i}_t\right)^2g^{A^i}_t
   +B^i_tq^i_t+Q^{q^i}_tg^{B^i}_t+g^{C^i}_t\right)dt +\left(\left(Q^{q^i}_t\right)^2Z^{A^i}_t+Q^{q^i}_tZ^{B^i}_t+Z^{C^i}_t\right)dW_t.
   \end{align*}
Therefore it holds
   \begin{align*}
   &dV^{q^i}_t+\left(-Q^{q^i}_t\left(\mu_t +a\sum_{j\neq i}q^j_t-\lambda_i\sigma^2_tQ^{q^i}_t\right) +b \left(q^i_t\right)^2\right)dt\\
   &=\left(2A^i_tQ^{q^i}_tq^i_t +\left(Q^{q^i}_t\right)^2g^{A^i}_t
   +B^i_tq^i_t+Q^{q^i}_tg^{B^i}_t+g^{C^i}_t-Q^{q^i}_t\left(\mu_t +a\sum_{j\neq i}q^j_t-\lambda_i\sigma^2_tQ^{q^i}_t\right) +b \left(q^i_t\right)^2\right)dt\\
   &\quad +\left(\left(Q^{q^i}_t\right)^2Z^{A^i}_t+Q^{q^i}_tZ^{B^i}_t+Z^{C^i}_t\right)dW_t.
   \end{align*}
and rearranging the drift terms, one can see
   \begin{align*}
   &dV^{q^i}_t+\left(-Q^{q^i}_t\left(\mu_t +a\sum_{j\neq i}q^j_t-\lambda_i\sigma^2_tQ^{q^i}_t\right) +b \left(q^i_t\right)^2\right)dt\\
   &=\left(2A^i_tQ^{q^i}_tq^i_t +\frac{1}{b}\left(Q^{q^i}_t\right)^2\left(A^i_t\right)^2
   +B^i_tq^i_t+\frac{1}{b}Q^{q^i}_tA^i_tB^i_t+\frac{1}{4b}\left(B^i_t\right)^2 +b \left(q^i_t\right)^2\right)dt\\
   &\quad +\left(\left(Q^{q^i}_t\right)^2Z^{A^i}_t+Q^{q^i}_tZ^{B^i}_t+Z^{C^i}_t\right)dW_t\\
   &=\frac{1}{b}\left(A^i_tQ^{q^i}_t+bq^i_t+\frac{1}{2}B^i_t\right)^2dt +\left(\left(Q^{q^i}_t\right)^2Z^{A^i}_t+Q^{q^i}_tZ^{B^i}_t+Z^{C^i}_t\right)dW_t.
   \end{align*}
Hence, it holds that
   \begin{align*}
   V^{q^i}_t&=V^{q^i}_T+\int_t^T\left(-Q^{q^i}_s\left(\mu_s +a\sum_{j\neq i}q^j_s-\lambda_i\sigma^2_tQ^{q^i}_s\right) +b \left(q^i_s\right)^2\right)ds\\
   &-\int_t^T\frac{1}{b}\left(A^i_sQ^{q^i}_s+bq^i_s+\frac{1}{2}B^i_s\right)^2ds -\int_t^T\left(\left(Q^{q^i}_s\right)^2Z^{A^i}_s+Q^{q^i}_sZ^{B^i}_s+Z^{C^i}_s\right)dW_s.
   \end{align*}
Therefore, for any $q^i,\bar{q}^i\in\mathcal{H}^2(\mathbb{R})$ and $t\in[0,T]$, by taking $\tilde{q}^i_s=q^i_s1_{\{s\leq t\}}+\bar{q}^i_s1_{\{s>t\}}$ for all $s\in[0,T]$, we have
   \begin{align*}
   V^{q^i}_t=V^{\tilde{q}^i}_t&=\beta_iQ^{\bar{q}^i}_{t,T}+\int_t^T\left(-Q^{\bar{q}^i}_{t,s}\left(\mu_s +a\sum_{j\neq i}q^j_s-\lambda_i\sigma^2_tQ^{\bar{q}^i}_{t,s}\right) +b \left(\bar{q}^i_s\right)^2\right)ds\\
   &-\int_t^T\frac{1}{b}\left(A^i_sQ^{\bar{q}^i}_{t,s}+b\bar{q}^i_s+\frac{1}{2}B^i_s\right)^2ds -\int_t^T\left(\left(Q^{\bar{q}^i}_{t,s}\right)^2Z^{A^i}_s+Q^{\bar{q}^i}_{t,s}Z^{B^i}_s+Z^{C^i}_s\right)dW_s
   \end{align*}
   which implies that
   \begin{align*}
   V^{q^i}_t\leq E\left[\beta_iQ^{\bar{q}^i}_{t,T}+\int_t^T\left(-Q^{\bar{q}^i}_{t,s}\left(\mu_s +a\sum_{j\neq i}q^j_s-\lambda_i\sigma^2_tQ^{\bar{q}^i}_{t,s}\right) +b \left(\bar{q}^i_s\right)^2\right)ds\Bigg|\mathcal{F}_t\right].
   \end{align*}
Hence, it holds that
   \begin{align*}
   V^{q^i}_t&\leq \essinf_{\bar{q}^i\in\mathcal{H}^2(\mathbb{R})}E\left[\beta_iQ^{\bar{q}^i}_{t,T}+\int_t^T\left(-Q^{\bar{q}^i}_{t,s}\left(\mu_s +a\sum_{j\neq i}q^j_s-\lambda_i\sigma^2_tQ^{\bar{q}^i}_{t,s}\right) +b \left(\bar{q}^i_s\right)^2\right)ds\Bigg|\mathcal{F}_t\right]\\
   &= \hat{\Phi}^i_t\left(Q^{q^i}_t\right).
   \end{align*}
   On the other hand, for any $t\in[0,T]$ and $q^i\in\mathcal{H}^2(\mathbb{R})$, the following random ODE
   \begin{equation*}
   Q^{i*}_s=Q^{q^i}_t-\frac{1}{b}\int_t^s\left(A^i_uQ^{i*}_u+\frac{1}{2}B^i_u\right)du,~~~s\in[t,T]
   \end{equation*}
   admits a unique solution $Q^{i*}\in \mathcal{S}^{2}(\mathbb{R})$ on $[t,T]$. Therefore, by taking $\tilde{q}^i_s=q^i_s1_{\{s\leq t\}}+q^{i*}_s1_{\{s>t\}}$ with
   \begin{equation*}
   q^{i*}_s=-\frac{1}{b}\left(A^i_sQ^{i*}_s+\frac{1}{2}B^i_s\right),
   \end{equation*}
   we have
   \begin{align*}
   V^{q^i}_t=V^{\tilde{q}^i}_t&=\beta_iQ^{q^{i*}}_{t,T}+\int_t^T\left(-Q^{q^{i*}}_{t,s}\left(\mu_s +a\sum_{j\neq i}q^j_s-\lambda_i\sigma^2_tQ^{q^{i*}}_{t,s}\right) +b \left(q^{i*}_s\right)^2\right)ds\\
   &-\int_t^T\frac{1}{b}\left(A^i_sQ^{q^{i*}}_{t,s}+bq^{i*}_s+\frac{1}{2}B^i_s\right)^2ds -\int_t^T\left(\left(Q^{q^{i*}}_{t,s}\right)^2Z^{A^i}_s+Q^{q^{i*}}_{t,s}Z^{B^i}_s+Z^{C^i}_s\right)dW_s
   \end{align*}
   which implies that
   \begin{align*}
   V^{q^i}_t&= E\left[\beta_iQ^{q^{i*}}_{t,T}+\int_t^T\left(-Q^{q^{i*}}_{t,s}\left(\mu_s +a\sum_{j\neq i}q^j_s-\lambda_i\sigma^2_tQ^{q^{i*}}_{t,s}\right) +b \left(q^{i*}_s\right)^2\right)ds\Bigg|\mathcal{F}_t\right]\\
   &\geq \hat{\Phi}^i_t\left(Q^{q^i}_t\right).
   \end{align*}
   Therefore, it holds that
      \begin{equation*}
   \hat{\Phi}^i_t\left(Q^{q^i}_t\right)= A^i_t\left(Q^{q^i}_t\right)^2
   +B^i_tQ^{q^i}_t+C^i_t.
   \end{equation*}
It is easy to verify that the unique optimal strategy $($feedback form$)$ for the agent $i$ is given by
\begin{equation*}
q^{i*}_t=-\frac{1}{b}\left(A^i_t Q^{q^{i*}}_t
   +\frac{1}{2}B^i_t\right).
\end{equation*}
\end{proof}
\subsection{Characterization of the optimal strategy in terms of an FBSDE}
In this section, we show that the optimal strategy for agent $i$ can be given by the unique solution of an FBSDE.
\begin{theorem}\label{thmi}
Suppose that $\beta_i\geq 0$ and $q^j\in\mathcal{H}^2(\mathbb{R}),~j\neq i$, then $\left(Q^{q^{i*}},q^{i*},\frac{Q^{q^{i*}}Z^{A^i}}{b}+\frac{Z^{B^i}}{2b}\right)$ is the unique solution of the following FBSDE
\begin{equation}\label{FBSDE_i}
\begin{cases}
\displaystyle &Q^{q^i}_t=Q^i_0+\int_0^tq^i_sds,\\
\displaystyle &q^i_t=-\frac{\beta_i}{b}Q^{q^i}_T+\int_t^T\frac{1}{b}\left(-\lambda_i\sigma^2_sQ^{q^i}_s+\frac{\left(\mu_s+a\sum_{j\neq i}q^j_s\right)}{2}\right)ds+\int_t^TZ^i_sdW_s.
\end{cases}
\end{equation}
in $\mathcal{S}^{2}(\mathbb{R})\times\mathcal{S}^2(\mathbb{R})\times\mathcal{H}^2(\mathbb{R}^{d})$.
\end{theorem}
\begin{proof}
By denoting
\begin{equation*}
\Lambda^i_t:=e^{-\int_0^t\frac{1}{b}A^i_sds},
\end{equation*}
it is easy to deduce that:
\begin{align*}
d\left(\Lambda^i_tQ^{q^{i*}}_tA^i_t\right)&=\Lambda^i_tQ^{q^{i*}}_tdA^i_t+\Lambda^i_tq^{i*}_tA^i_tdt+Q^{q^{i*}}_tA^i_td\Lambda^i_t\\
&=\Lambda^i_tq^{i*}_tA^i_tdt+\frac{\Lambda^i_tQ^{q^{i*}}_t\left(A^i_t\right)^2}{b}dt-\lambda_i\sigma^2_t\Lambda^i_tQ^{q^{i*}}_tdt-\frac{\Lambda^i_tQ^{q^{i*}}_t\left(A^i_t\right)^2}{b}dt+\Lambda^i_tQ^{q^{i*}}_tZ^{A^i}_tdW_t\\
&=\Lambda^i_tq^{i*}_tA^i_tdt-\lambda_i\sigma^2_t\Lambda^i_tQ^{q^{i*}}_tdt+\Lambda^i_tQ^{q^{i*}}_tZ^{A^i}_tdW_t.
\end{align*}
Therefore, it holds that
\begin{align*}
\Lambda^i_tQ^{q^{i*}}_tA^i_t&=\beta_i\Lambda^i_TQ^{q^{i*}}_T-\int_t^T\Lambda^i_sq^{i*}_sA^i_sds+\int_t^T\lambda_i\sigma^2_s\Lambda^i_sQ^{q^{i*}}_sds-\int_t^T\Lambda^i_sQ^{q^{i*}}_sZ^{A^i}_sdW_s.
\end{align*}
Noting that
\begin{equation*}
\Lambda^i_tB^i_t=-\int_t^T\Lambda^i_s\left(\mu_s+a\sum_{j\neq i}q^j_s\right)ds-\int_t^T\Lambda^i_sZ^{B^i}_sdW_s,
\end{equation*}
one has
\begin{align*}
\Lambda^i_tq^{i*}_t&=-\frac{\beta_i}{b}\Lambda^i_TQ^{q^{i*}}_T+\int_t^T\frac{1}{b}\left(\Lambda^i_sq^{i*}_sA^{i}_s-\lambda_i\sigma^2_s\Lambda^i_sQ^{q^{i*}}_s+\frac{\Lambda^i_s\left(\mu_s+a\sum_{j\neq i}q^j_s\right)}{2}\right)ds\\
&\quad+\int_t^T\left(\frac{\Lambda^i_sQ^{q^{i*}}_sZ^{A^i}_s}{b}+\frac{\Lambda^i_sZ^{B^i}_s}{2b}\right)dW_s.
\end{align*}
Therefore, it holds that
\begin{align*}
dq^{i*}_t &= d\left(\left(\Lambda^i_t\right)^{-1}\Lambda^i_tq^{i*}_t\right)\\
&= \frac{A^i_t}{b}q^{i*}_tdt - \left(\Lambda^i_t\right)^{-1}\left(\frac{1}{b}\left(\Lambda^i_tq^{i*}_tA^{i}_t-\lambda_i\sigma^2_t\Lambda^i_tQ^{q^{i*}}_t+\frac{\Lambda^i_t\left(\mu_t+a\sum_{j\neq i}q^j_t\right)}{2}\right)dt\right.\\
&\quad\quad\quad\left.+\left(\frac{\Lambda^i_tQ^{q^{i*}}_tZ^{A^i}_t}{b}+\frac{\Lambda^i_tZ^{B^i}_t}{2b}\right)dW_t\right)\\
&= \frac{1}{b}\left(\lambda_i\sigma^2_tQ^{q^{i*}}_t-\frac{\left(\mu_t+a\sum_{j\neq i}q^j_t\right)}{2}\right)dt-\left(\frac{Q^{q^{i*}}_tZ^{A^i}_t}{b}+\frac{Z^{B^i}_t}{2b}\right)dW_t.
\end{align*}
It is easy to check that $\left(Q^{q^{i*}},q^{i*},\frac{Q^{q^{i*}}Z^{A^i}}{b}+\frac{Z^{B^i}}{2b}\right)$ is in $\mathcal{S}^{2}(\mathbb{R})\times\mathcal{S}^2(\mathbb{R})\times\mathcal{H}^2(\mathbb{R}^{d})$. We now prove the uniqueness. Suppose that FBSDE \eqref{FBSDE_i} admits another solution $(Q^{\bar{q}^i},\bar{q}^i,\bar{Z}^i)\in\mathcal{S}^2(\mathbb{R})\times\mathcal{S}^2(\mathbb{R})\times\mathcal{H}^2(\mathbb{R}^{d})$. Then, we have
 \begin{equation*}
\begin{cases}
\displaystyle &Q^{q^i}_t-Q^{\bar{q}^i}_t=\int_0^t\left(q^i_s-\bar{q}^i_s\right)ds,\\
\displaystyle &q^i_t-\bar{q}^i_t=-\frac{\beta_i}{b}\left(Q^{q^i}_T-Q^{\bar{q}^i}_T\right)+\int_t^T\frac{1}{b}\left(-\lambda_i\sigma^2_s\left(Q^{q^i}_s-Q^{\bar{q}^i}_s\right)\right)ds+\int_t^T\left(Z^i_s-\bar{Z}^i_s\right)dW_s.
\end{cases}
\end{equation*}
Therefore, it holds that
\begin{align*}
\left(q^i_t-\bar{q}^i_t\right)\left(Q^{q^i}_t-Q^{\bar{q}^i}_t\right)&=-\frac{\beta_i}{b}\left(Q^{q^i}_T-Q^{\bar{q}^i}_T\right)^2+\int_t^T\frac{1}{b}\left(-\lambda_i\sigma^2_s\left(Q^{q^i}_s-Q^{\bar{q}^i}_s\right)^2\right)ds\\
&\quad\quad -\int_t^T\left(q^i_s-\bar{q}^i_s\right)^2ds+\int_t^T\left(Q^{q^i}_t-Q^{\bar{q}^i}_t\right)\left(Z^i_s-\bar{Z}^i_s\right)dW_s.
\end{align*}
Thus, it holds that
\begin{align*}
0=E\left[-\frac{\beta_i}{b}\left(Q^{q^i}_T-Q^{\bar{q}^i}_T\right)^2-\int_0^T\frac{\lambda_i\sigma^2_s}{b}\left(Q^{q^i}_s-Q^{\bar{q}^i}_s\right)^2ds -\int_0^T\left(q^i_s-\bar{q}^i_s\right)^2ds\right]\leq 0
\end{align*}
which implies uniqueness.
\end{proof}

\section{Characterization and existence of a Nash equilibrium}
We first provide a characterizing result of a Nash equilibrium in terms of a system of FBSDE.
\begin{theorem}\label{thm 4.1}
Suppose that $\beta_i\geq 0$, if the following FBSDE:
\begin{equation}\label{FBSDE}
\begin{cases}
\displaystyle &Q^{q^i}_t=Q^i_0+\int_0^tq^i_sds,~~i=1,\ldots,n\\
\displaystyle &q^i_t=-\frac{\beta_i}{b}Q^{q^i}_T+\int_t^T\frac{1}{b}\left(-\lambda_i\sigma^2_sQ^{q^i}_s+\frac{\left(\mu_s+a\sum_{j\neq i}q^j_s\right)}{2}\right)ds+\int_t^TZ^i_sdW_s,~~i=1,\ldots,n.
\end{cases}
\end{equation}
admits a solution $(Q^{q},q,Z)\in\mathcal{S}^{2}(\mathbb{R}^n)\times\mathcal{S}^2(\mathbb{R}^n)\times\mathcal{H}^2(\mathbb{R}^{n\times d})$, then $q$ is a Nash equilibrium. On the other hand, if $q\in\mathcal{H}^2(\mathbb{R}^n)$ is a Nash equilibrium, then $(Q^q,q,Z)$ is a solution of FBSDE \ref{FBSDE} in $\mathcal{S}^{2}(\mathbb{R}^n)\times\mathcal{S}^2(\mathbb{R}^n)\times\mathcal{H}^2(\mathbb{R}^{n\times d})$, where $Z$ is given by
\begin{equation*}
Z=\left(\frac{Q^{q^{1}}Z^{A^1}}{b}+\frac{Z^{B^1}}{2b},\ldots,\frac{Q^{q^{n}}Z^{A^n}}{b}+\frac{Z^{B^n}}{2b}\right)^{'},
\end{equation*}
where for $i=1,\ldots,n$, $Z^{A^i},Z^{B^i}$ are given as in Theorem \ref{Thm1} and $M^{'}$ denotes the transpose of the matrix $M$.
\end{theorem}
\begin{proof}
The result follows directly from Theorem \ref{Thmi} and  Theorem \ref{thmi}.
\end{proof}
In order to get a Nash equilibrium, it is sufficient to have the existence of solution for FBSDE \eqref{FBSDE}. In this section, we will investigate the solvability for FBSDE \eqref{FBSDE}. An existence and uniqueness result for small time horizon is due to Antonelli \cite{Ant}. Under some assumptions, we get a unique global solution for FBSDE \eqref{FBSDE} which is stated in the following theorem.
\begin{theorem} \label{thm 4.2}
Suppose that $\beta_i>0$, $\lambda_i>0$ and $\lambda_i\sigma^2_t>\frac1{16}a^2b(n-1)$ for all $i=1,\ldots,n$ and $t\in[0,T]$, then FBSDE \eqref{FBSDE} admits a unique solution $(Q^q,q,Z)\in\mathcal{S}^2(\mathbb{R}^n)\times\mathcal{S}^2(\mathbb{R}^n)\times\mathcal{H}^2(\mathbb{R}^{n\times d})$.
\end{theorem}
\begin{proof}
Denoting $\tilde{q}^i_t=-q^i_t$, we have
\begin{equation*}
\begin{cases}
\displaystyle &Q^{q^i}_t=Q^i_0-\int_0^t\tilde{q}^i_sds, ~~ i=1,\ldots,n\\
\displaystyle &\tilde{q}^i_t=\frac{\beta_i}{b}Q^{q^i}_T-\int_t^T\frac{1}{b}\left(-\lambda_i\sigma^2_sQ^{q^i}_s+\frac{\mu_s}{2}-\frac{a\sum_{j\neq i}\tilde{q}^j_s}{2}\right)ds-\int_t^TZ^i_sdW_s, ~~i=1,\ldots,n.
\end{cases}
\end{equation*}
Since it holds that
\begin{align*}
\sum_{i=1}^n\frac{\beta_i}{b}|x_i|^2\geq \min_{1\leq i\leq n}\frac{\beta_i}{b}|x|^2
\end{align*}
and
\begin{align*}
\sum_{i=1}^n\left(-\frac{\lambda_i\sigma^2_t}{b}|x_i|^2-|y_i|^2-\frac{a}{2}\sum_{j\neq i}y_jx_i\right)&\leq \sum_{i=1}^n\left(-\frac{\lambda_i\sigma^2_t}{b}|x_i|^2-|y_i|^2+\frac{a^2(n-1)}{16}|x_i|^2+\sum_{j\neq i}\frac{1}{n-1}|y_j|^2\right)\\
&= \sum_{i=1}^n\left(-\frac{\lambda_i\sigma^2_t}{b}+\frac{a^2(n-1)}{16}\right)|x_i|^2\\
&\leq - \min_{1\leq i\leq n}\inf_{0\leq t\leq T}\left(\frac{\lambda_i\sigma^2_t}{b}-\frac{a^2(n-1)}{16}\right)|x|^2,
\end{align*}
the monotonicity condition in Peng-Wu \cite{PW} is satisfied. Therefore, the solvability follows.
\end{proof}

As a direct consequence of Theorem \ref{thm 4.1} and Theorem \ref{thm 4.2}, we have the following corollary on the existence and uniqueness of a Nash equilibrium.
\begin{corollary}
Suppose that $\beta_i>0$, $\lambda_i>0$ and $\lambda_i\sigma^2_t>\frac1{16}a^2b(n-1)$ for all $i=1,\ldots,n$ and $t\in[0,T]$, then there exists a unique Nash equilibrium.
\end{corollary}

\subsection{A Riccati-type equation}
Since FBSDE \eqref{FBSDE} is linear, we will investigate it's solvability through Riccati equations. Indeed, FBSDE \eqref{FBSDE} could be rewritten as
\begin{equation*}
\begin{cases}
\displaystyle &Q^{q}_t=Q_0+\int_0^tq_sds,\\
\displaystyle &q(t)=GQ^{q}_T+\int_t^T\left(\hat{A}_sQ^{q}_s+(-\frac{a}{2b}I_n+\frac{a}{2b}\hat{B})q_s+\hat{C}_s\right)ds+\int_t^TZ_sdW_s
\end{cases}
\end{equation*}
where $G$ is $n\times n$ diagonal matrix with diagonal elements $-\frac{\beta_i}{b}$, $\hat{A}_s$ is $n\times n$ diagonal matrix with diagonal elements $-\frac{\lambda_i\sigma^2_s}{b}$, $I_n$ is $n\times n$ identity matrix, $\hat{B}$ is $n\times n$ matrix whose elements are all equal to $1$ and $\hat{C}_s=\left(\frac{\mu_s}{2b},\ldots,\frac{\mu_s}{2b}\right)^T$.

Suppose that the following holds:
\begin{equation*}
q_t = P_tQ^q_t+p_t, ~~t\in[0,T],
\end{equation*}
with $(P,\Lambda)$ and $(p,\eta)$ being the adapted solutions of the following BSDEs respectively:
\begin{equation*}
\begin{cases}
&dP_t=\Gamma_tdt+\Lambda_tdW_t,\\
&P_T=G
\end{cases}
\end{equation*}
and
\begin{equation*}
\begin{cases}
&dp_t=\xi_tdt+\eta_tdW_t,\\
&p_T=0
\end{cases}
\end{equation*}
where $\Gamma$ and $\xi$ will be chosen later. Applying It\^{o}'s formula, we have the following
\begin{align*}
&\left(\Gamma_tQ^q_t+P_tq_t+\xi_t\right)dt+\left(\Lambda_tQ^q_t+\eta_t\right)dW_t\\
&=dq_t=-\left(\hat{A}_tQ^q_t+\left(-\frac{a}{2b}I_n+\frac{a}{2b}\hat{B}\right)q_t+\hat{C}_t\right)dt-Z_tdW_t.
\end{align*}
Comparing drift and diffusion terms, we should have
\begin{equation*}
\begin{cases}
&\left(\Gamma_t+P^2_t\right)Q^q_t+P_tp_t+\xi_t=-\left(\hat{A}_t+\left(-\frac{a}{2b}I_n+\frac{a}{2b}\hat{B}\right)P_t\right)Q^q_t-\left(\left(-\frac{a}{2b}I_n+\frac{a}{2b}\hat{B}\right)p_t+\hat{C}_t\right),\\
&\Lambda_tQ^q_t+\eta_t=-Z_t.
\end{cases}
\end{equation*}
Therefore, we will take
\begin{equation*}
\begin{cases}
&\Gamma = - \hat{A}_t-\left(-\frac{a}{2b}I_n+\frac{a}{2b}\hat{B}\right)P_t-P^2_t\\
&\xi_t = -\left(-\frac{a}{2b}I_n+\frac{a}{2b}\hat{B}+P_t\right)p_t-\hat{C}_t
\end{cases}
\end{equation*}
Thus, we obtain the following result.
\begin{proposition}\label{Ric}
Suppose that the following BSDE
\begin{equation}\label{Riccati}
\begin{cases}
&dP_t = \left(- \hat{A}_t-\left(-\frac{a}{2b}I_n+\frac{a}{2b}\hat{B}\right)P_t-P^2_t\right)dt + \Lambda_tdW_t,\\
&P_T=G
\end{cases}
\end{equation}
admits an adapted solution $(P,\Lambda)\in\mathcal{S}^m(\mathbb{R}^{n\times n})\times\mathcal{H}^m(\mathbb{R}^{n\times n\times d})$ for all $m\geq 1$. Suppose moreover that the following BSDE
\begin{equation}\label{Riccati1}
\begin{cases}
&dp_t=\left(-\left(-\frac{a}{2b}I_n+\frac{a}{2b}\hat{B}+P_t\right)p_t-\hat{C}_t\right)dt+\eta_tdW_t,\\
&p_T=0
\end{cases}
\end{equation}
admits a unique adapted solution $(p,\eta)\in\mathcal{S}^m(\mathbb{R}^{n})\times\mathcal{H}^m(\mathbb{R}^{n\times d})$ for all $m\geq 1$ and that the unique solution of following random ODE,
\begin{equation}\label{SDE}
Q^q_t=Q_0+\int_0^t\left(P_sQ^q_s+p_s\right)ds
\end{equation}
belongs to $\mathcal{S}^m(\mathbb{R}^{n})$ for all $m\geq 1$. Then FBSDE \eqref{FBSDE} admits an adapted solution $(Q^q,q,Z)\in\mathcal{S}^2(\mathbb{R}^n)\times\mathcal{S}^2(\mathbb{R}^n)\times\mathcal{H}^2(\mathbb{R}^{n\times d})$ such that $q_t = P_tQ^q_t+p_t$ and $Z_t=-\Lambda_tQ^q_t-\eta_t$.
\end{proposition}

\subsection{When $\sigma_t=\sigma$ and $\mu_t=\mu$}
In the current case, the FBSDE \eqref{FBSDE} takes the following form:
\begin{equation*}
\begin{cases}
\displaystyle &Q^{q}_t=Q_0+\int_0^tq_sds,\\
\displaystyle &q_t=GQ^{q}_T+\int_t^T\left(\hat{A}Q^{q}_s+(-\frac{a}{2b}I_n+\frac{a}{2b}\hat{B})q_s+\hat{C}\right)ds
\end{cases}
\end{equation*}
where $G$ is $n\times n$ diagonal matrix with diagonal elements $-\frac{\beta_i}{b}$, $\hat{A}$ is $n\times n$ diagonal matrix with diagonal elements $-\frac{\lambda_i\sigma^2}{b}$, $I_n$ is $n\times n$ identity matrix, $\hat{B}$ is $n\times n$ matrix whose elements are all equal to $1$ and $\hat{C}=\left(\frac{\mu}{2b},\ldots,\frac{\mu}{2b}\right)^T$.
By denoting $\tilde{A}=\frac{a}{2b}\left(I_n-\hat{B}\right)$, we obtain the following equivalent second order inhomogeneous ODE
\begin{equation*}
    Q^{\prime \prime} = \tilde{A} Q^\prime - \hat{A} Q -\hat{C} \quad \text{equivalent to}\quad\Lambda^\prime = M \Lambda + N
\end{equation*}
where
\begin{align*}
    \Lambda & =
    \begin{bmatrix}
        Q^\prime\\
        Q
    \end{bmatrix}
            &
    M & =
    \begin{bmatrix}
        \tilde{A} & -\hat{A} \\
        I_n & 0
    \end{bmatrix}
      &
    N & =
    \begin{bmatrix}
        -\hat{C}\\
        0
    \end{bmatrix}
\end{align*}
Since $M$ is invertible, the solution is given by
\begin{equation*}
    \Lambda = \exp\left( t M \right)
    \begin{bmatrix}
        \xi_1\\
        \xi_2
    \end{bmatrix}
    +\int_{0}^{t}\exp\left( s M \right) N ds = \exp\left( t M \right)
    \begin{bmatrix}
        \xi_1\\
        \xi_2
    \end{bmatrix}
    +\left( \exp\left( tM \right) - I_{2n} \right)M^{-1} N
\end{equation*}
where $\xi_1, \xi_2$ in $\mathbb{R}^{2n}$ is a vector to be determined by the conditions:
\begin{equation*}
    \Lambda[n+1, 2n](0)=Q(0) = Q_0, \quad \text{and}\quad \Lambda[1, n](T) = G Q(T)
\end{equation*}
It follows that $\xi_2 = Q_0$.
Hence the second condition is given by
\begin{multline*}
    \exp\left( T M \right)
    \begin{bmatrix}
        \xi_1\\
        Q_0
    \end{bmatrix}
    +\left( \exp\left( TM \right) - I_{2n} \right) M^{-1}N \\
    =
    \begin{bmatrix}
        0 & G \\
        0 & I_n
    \end{bmatrix}
    \left[
    \exp\left( T M \right)
    \begin{bmatrix}
        \xi_1\\
        Q_0
    \end{bmatrix}
    +\left( \exp\left( TM \right) - I_{2n} \right)M^{-1} N \right]
\end{multline*}
equivalent to
\begin{equation*}
    \begin{bmatrix}
        I_n & -G
    \end{bmatrix}
    \exp\left( T M \right)
    \begin{bmatrix}
        \xi_1\\
        Q_0
    \end{bmatrix}
    =
    \begin{bmatrix}
        -I_n & G
    \end{bmatrix}
    \left( \exp\left( TM \right) - I_{2n} \right)M^{-1} N
\end{equation*}
Hence, denoting by
\begin{equation*}
    \exp\left( TM \right)=
    \begin{bmatrix}
        E_1 & E_2\\
        E_3 & E_4
    \end{bmatrix}
\end{equation*}
it follows that $\xi_1$ is given by
\begin{equation*}
    \xi_1 = \left(E_1 - GE_3 \right)^{-1}
    \left[
        \begin{bmatrix}
            -I_n & G
        \end{bmatrix}
        \left( \exp\left( TM \right) - I_{2n} \right)M^{-1} N - \left(E_2 - GE_4 \right) Q_0
    \right].
\end{equation*}
\subsubsection{Numerical results}
Throughout we consider the following set of parameters

\begin{itemize}
    \item Market parameters
        \begin{itemize}
            \item drift: $\mu = 2\%$
            \item vol: $\sigma = 20\%$
            \item Maturity: $T=1$
            \item price impact: $a = 1\%$
            \item slippage: $b = 1\%$
        \end{itemize}
    \item 3 traders:
        \begin{itemize}
            \item Risk aversion: $\alpha = (1, 0.5, 0.25)$, $\lambda = (1, 0.5, 0.25) $,
            \item Position to liquidate: $Q = (1, 1, 0.5)$
        \end{itemize}
\end{itemize}
\begin{itemize}
    \item Dependence on the drift
        \begin{figure}[H]
            \centering
            \subfigure{\includegraphics[width=0.32\textwidth]{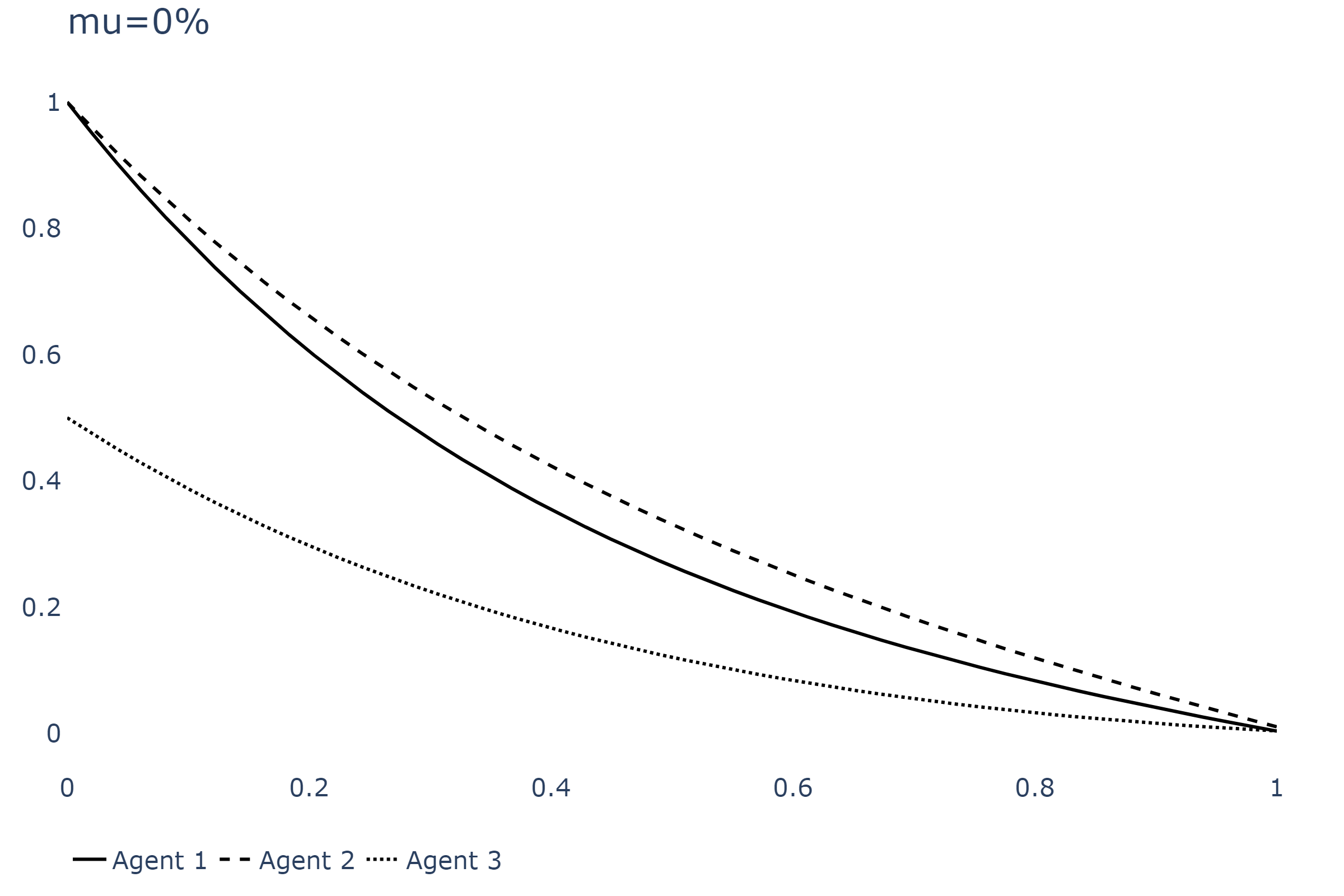}}
            \subfigure{\includegraphics[width=0.32\textwidth]{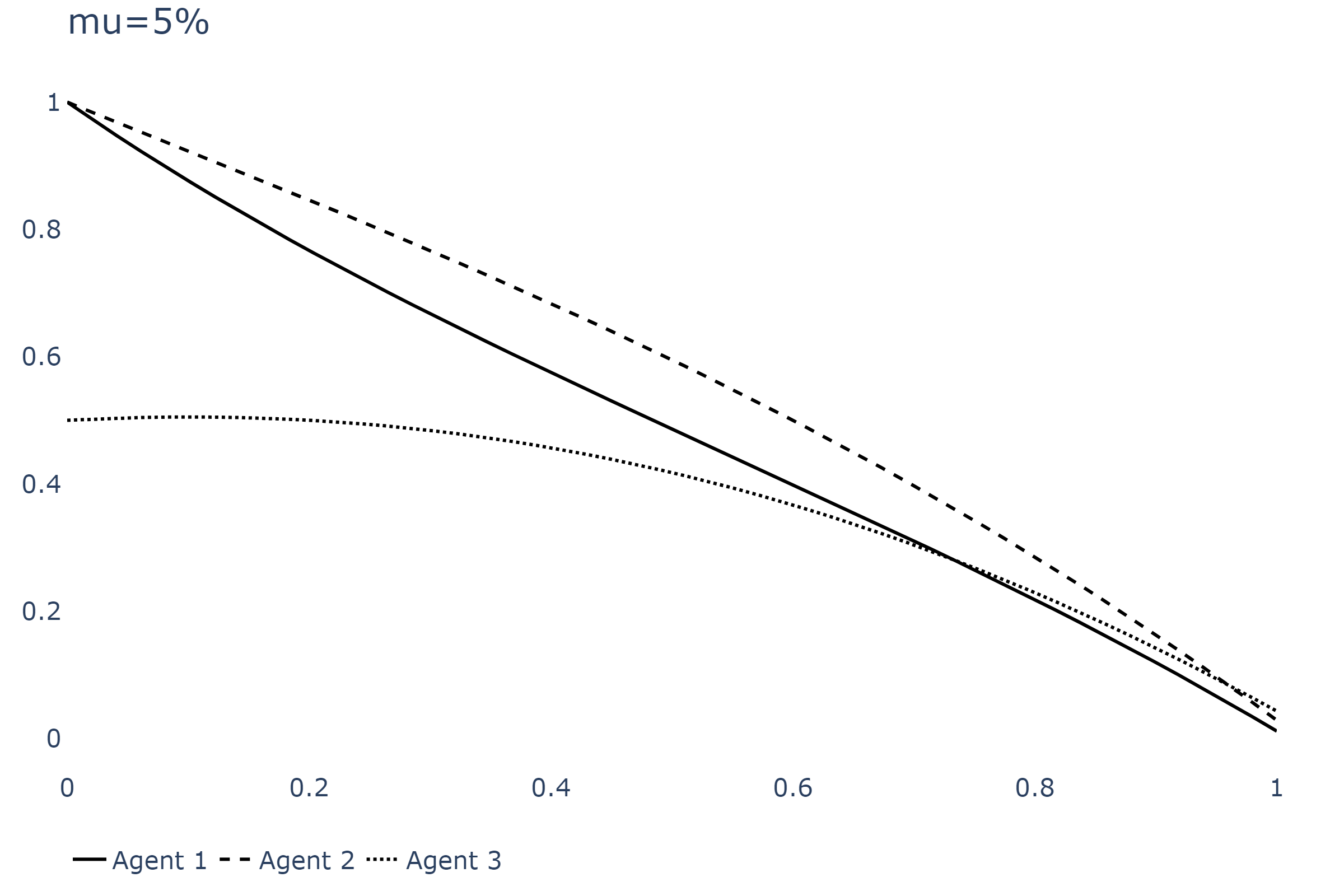}}
            \subfigure{\includegraphics[width=0.32\textwidth]{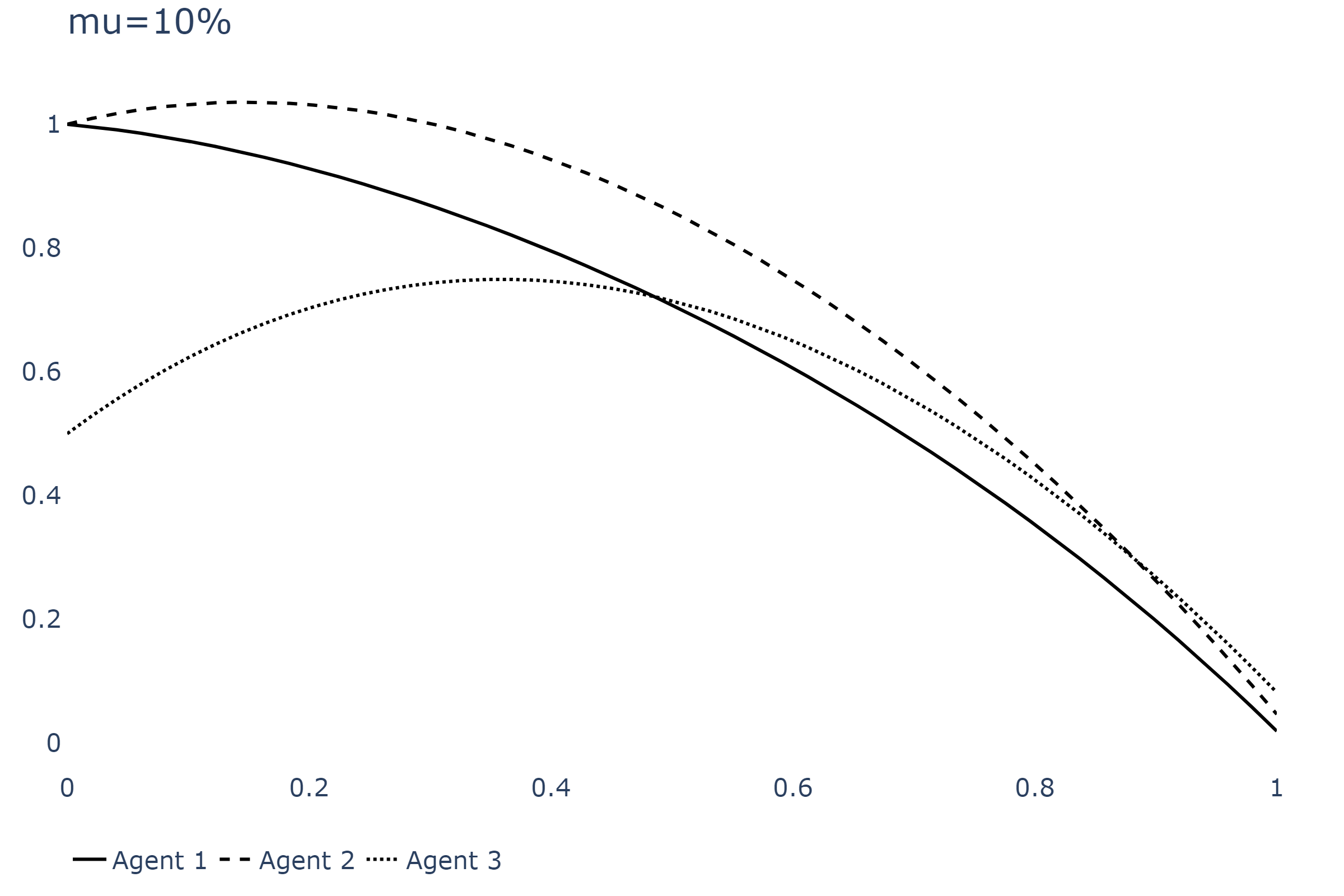}}
            \caption{Plot of the three agents'inventory for different drift values}
            \label{fig:mudep}
        \end{figure}
        As the drift increases, the agents tend to liquidate slowly or even start buying at the beginning to benefit from the future mean return which will compensate to the liquidation cost.
    \item Dependence on the volatility
        \begin{figure}[H]
            \centering
            \subfigure{\includegraphics[width=0.32\textwidth]{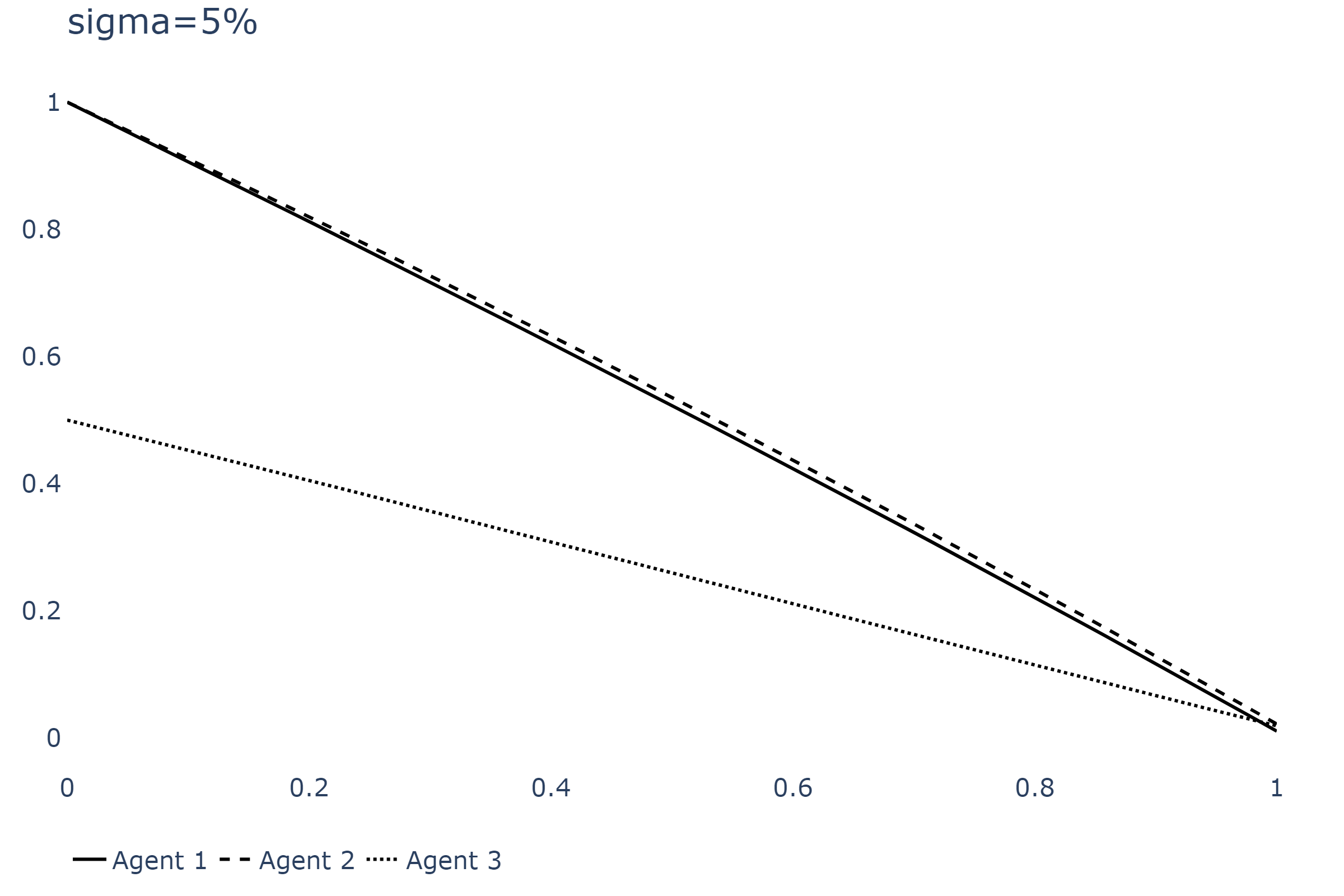}}
            \subfigure{\includegraphics[width=0.32\textwidth]{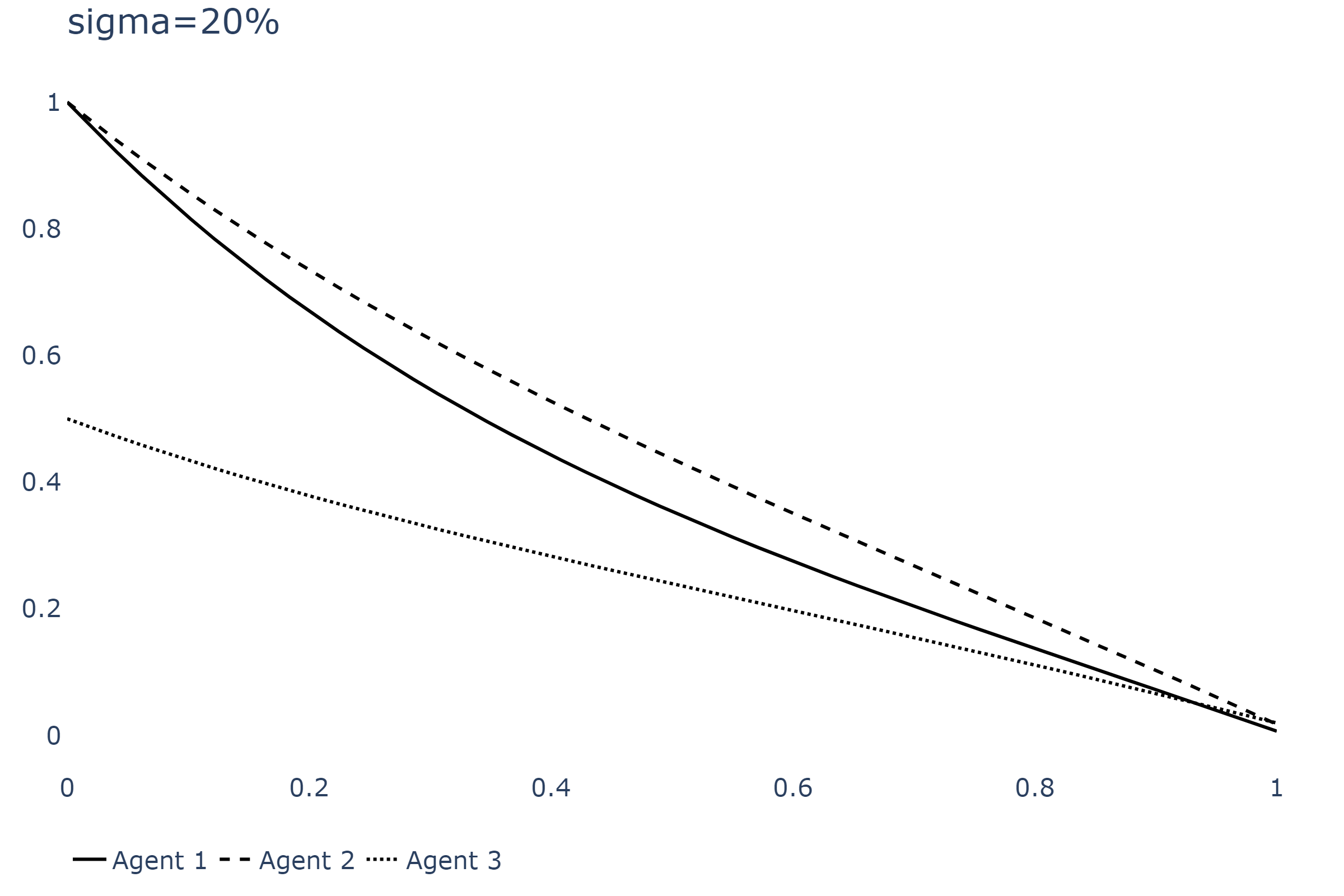}}
            \subfigure{\includegraphics[width=0.32\textwidth]{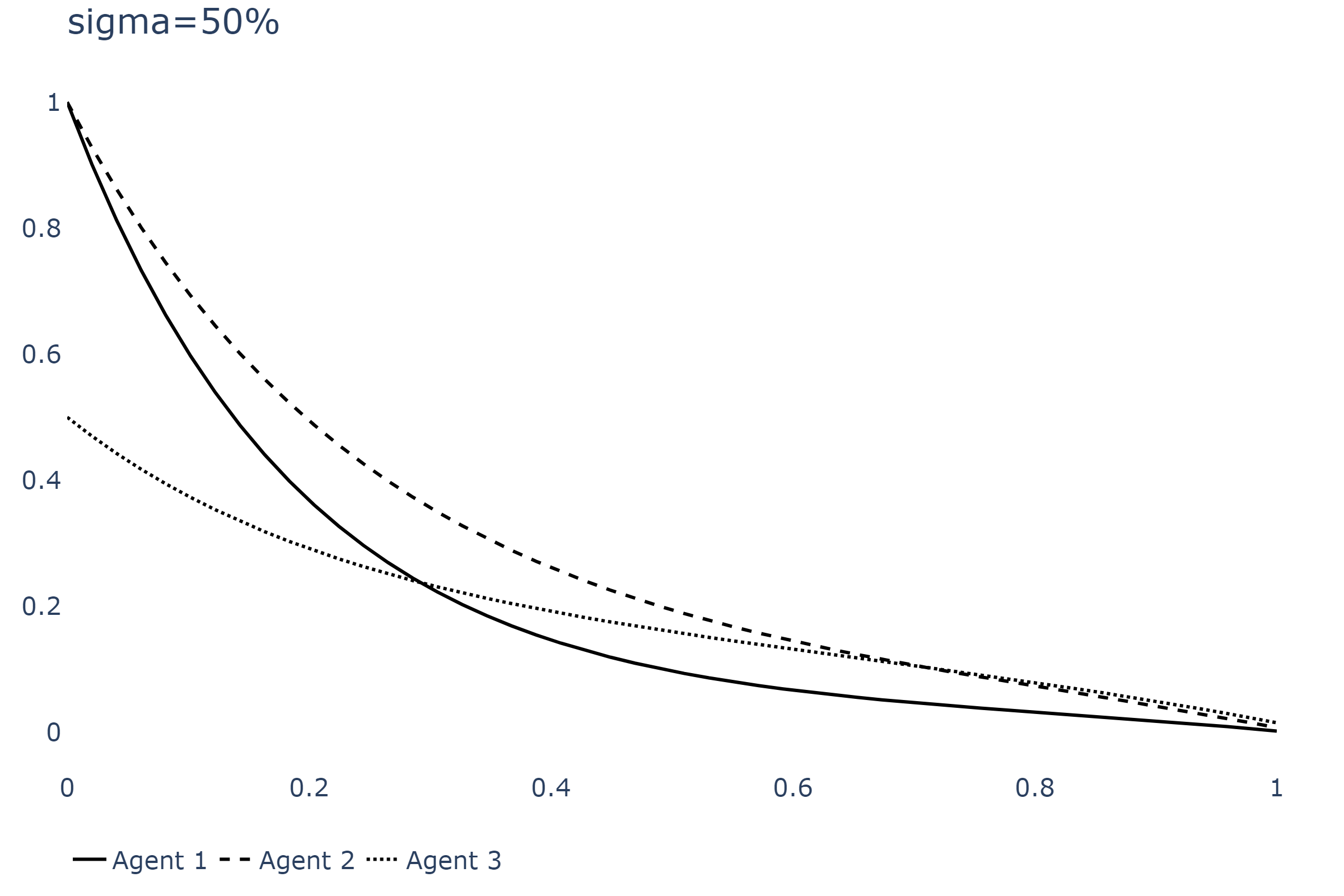}}
            \caption{Plot of the three agents'inventory for different volatility values}
            \label{fig:sigdep}
        \end{figure}
        As the volatility increases, the agents tend to liquidate quickly at the beginning to reduce the liquidation risk.
    \item Dependence on $a$
        \begin{figure}[H]
            \centering
            \subfigure{\includegraphics[width=0.32\textwidth]{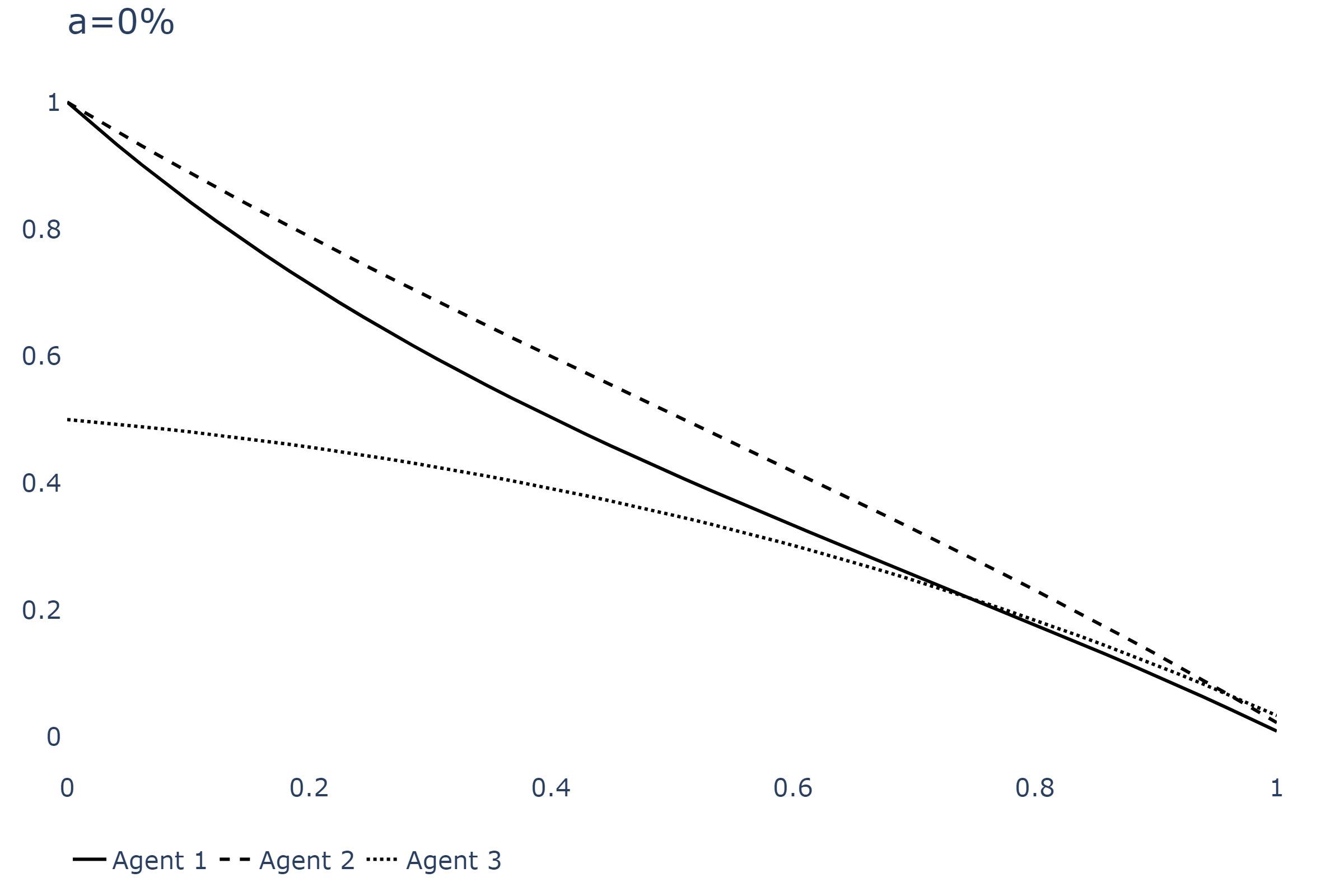}}
            \subfigure{\includegraphics[width=0.32\textwidth]{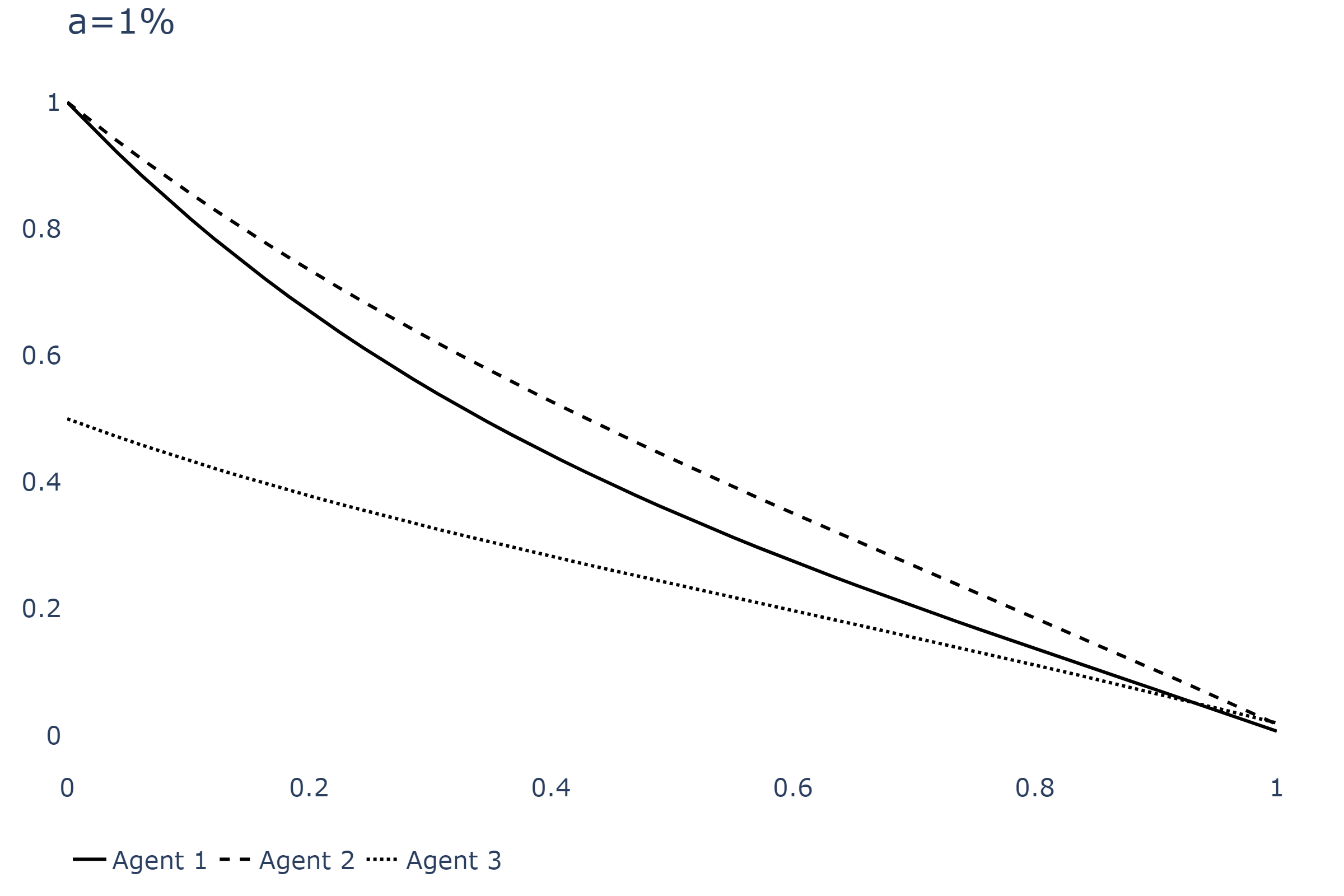}}
            \subfigure{\includegraphics[width=0.32\textwidth]{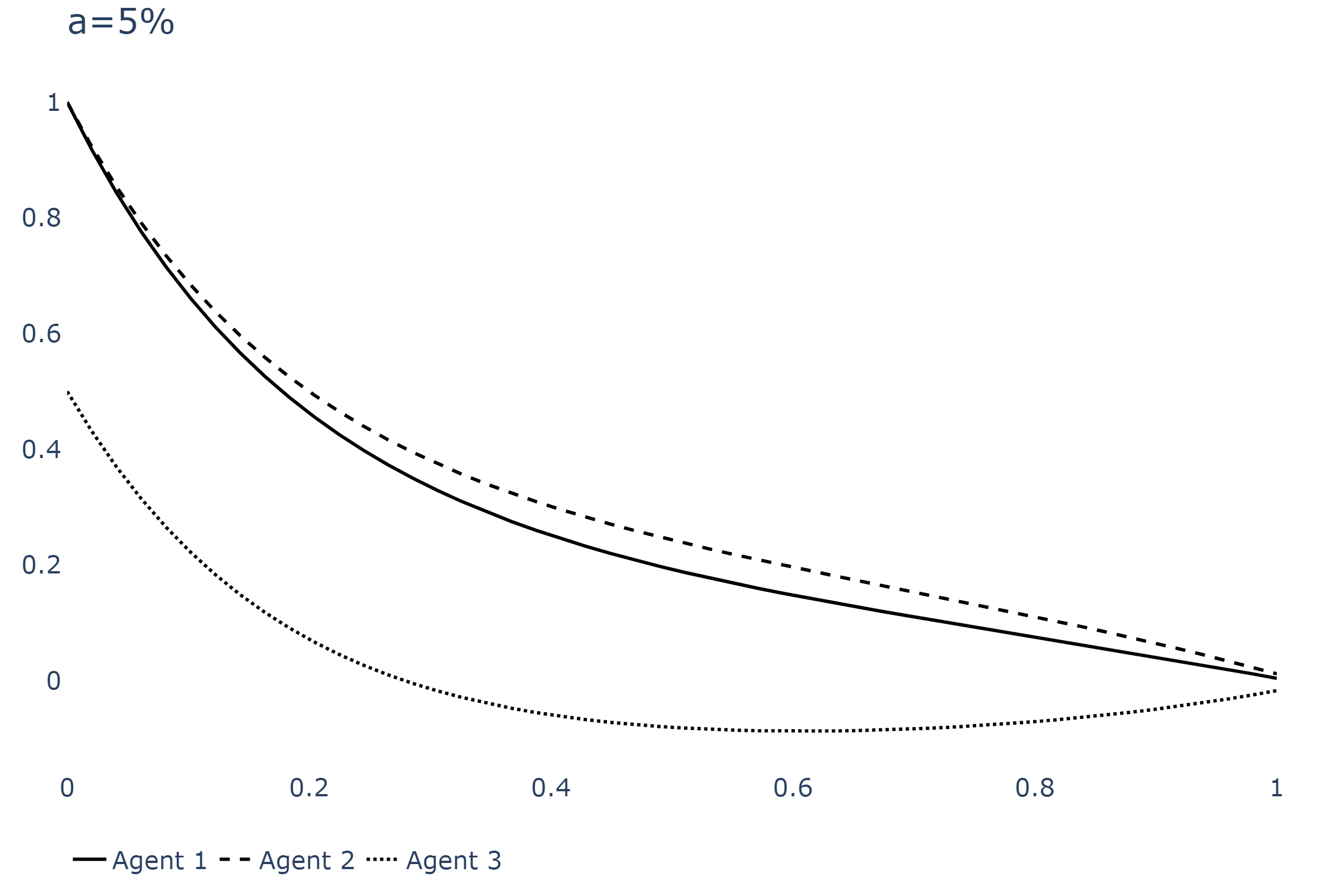}}
            \caption{Plot of the three agents'inventory for different values of price impact}
            \label{fig:adep}
        \end{figure}
        As the permanent market impact increases, the agents tend to liquidate quickly at the beginning to reduce the liquidation cost. For high permanent market impact, the agent with smaller initial inventory tend to short sell and reliquidate to make profit which will compensate to the liquidation cost.
    \item Dependence on $b$
        \begin{figure}[H]
            \centering
            \subfigure{\includegraphics[width=0.32\textwidth]{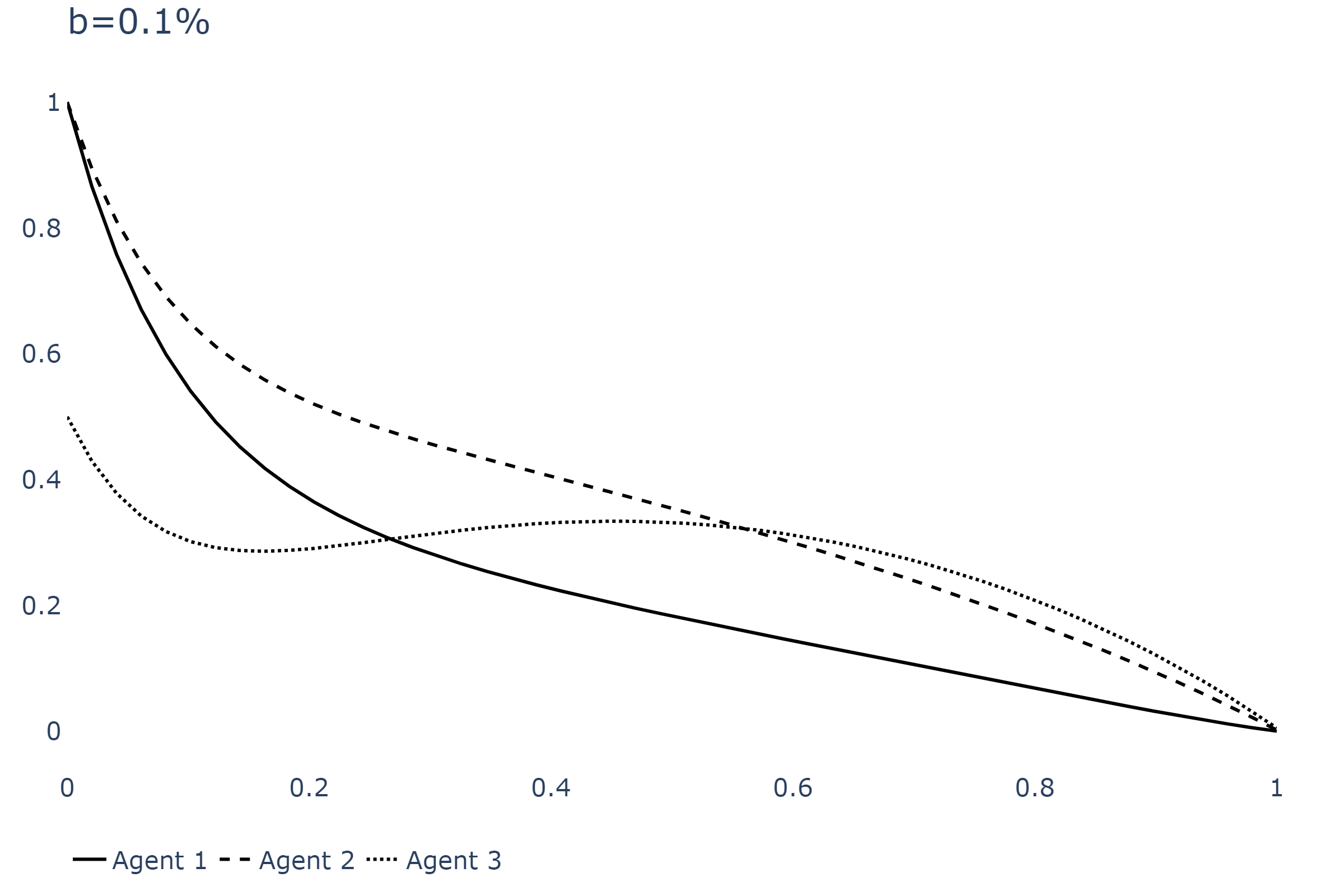}}
            \subfigure{\includegraphics[width=0.32\textwidth]{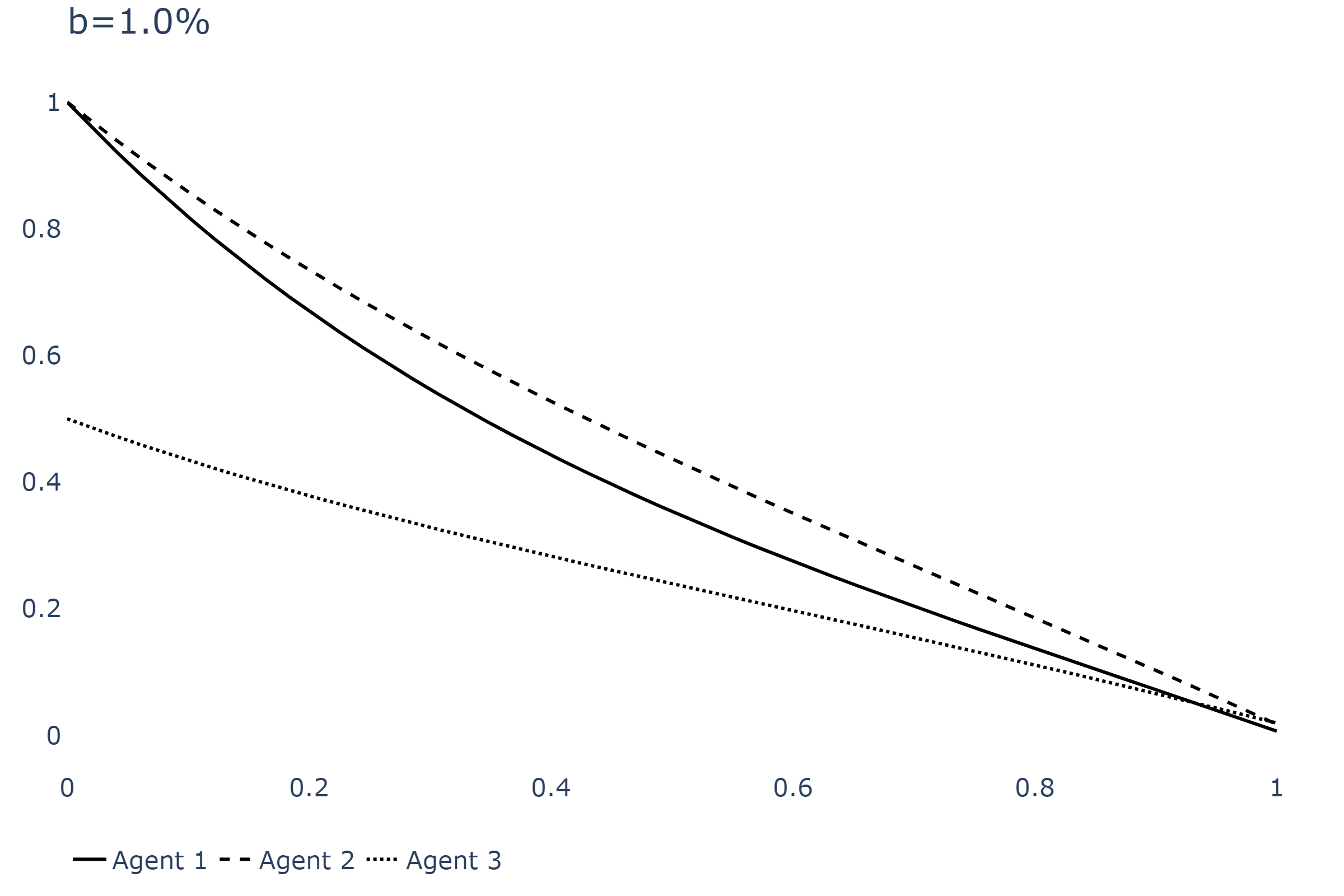}}
            \subfigure{\includegraphics[width=0.32\textwidth]{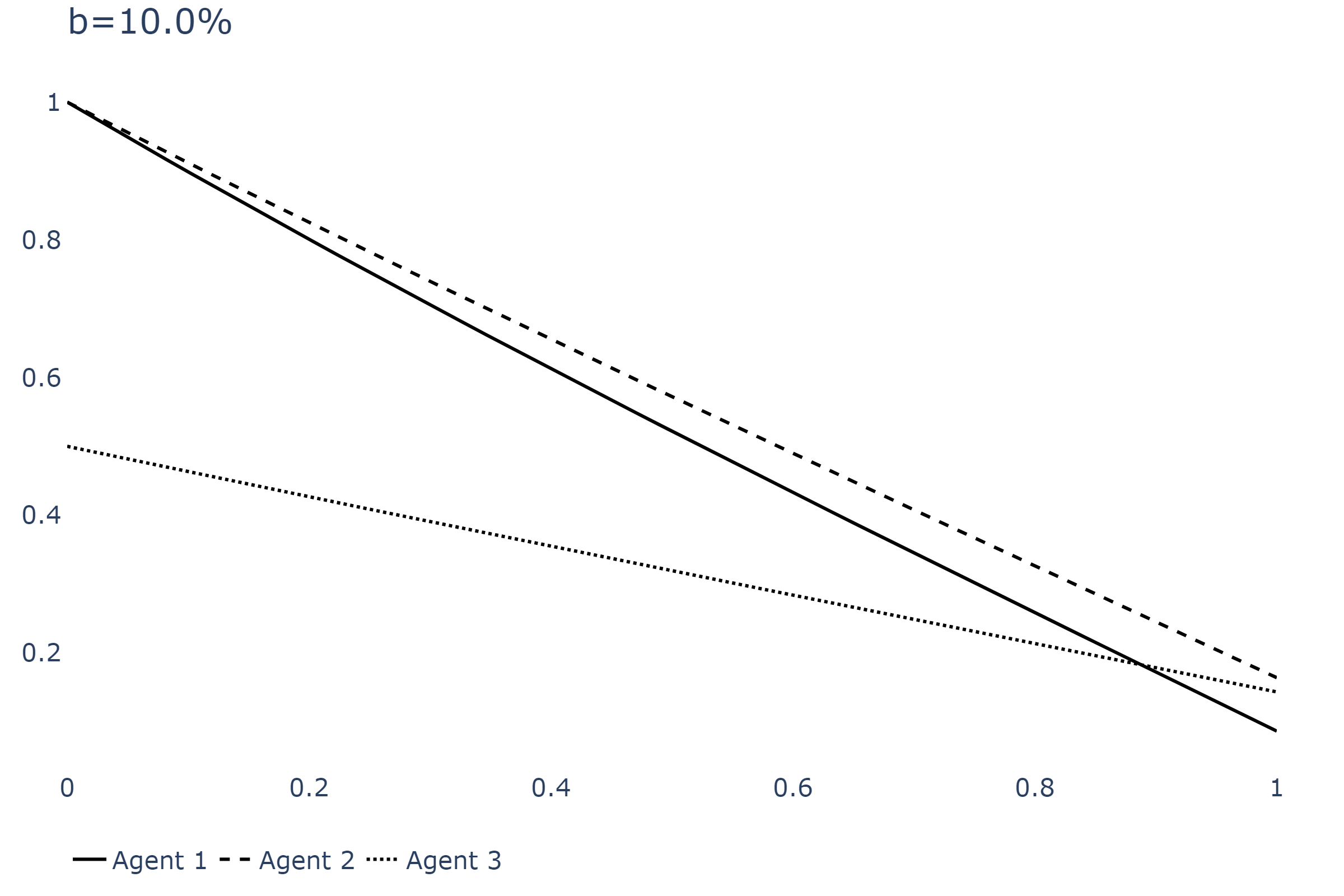}}
            \caption{Plot of the three agents'inventory for different values of slippage effect}
            \label{fig:bdep}
        \end{figure}
        For small slippage, the agent with smaller initial inventory tend to vary between liquidation and purchasing to make profit which will compensate to the liquidation cost.
    \item Dependence on $\alpha$ joint magnitude
        \begin{figure}[H]
            \centering
            \subfigure{\includegraphics[width=0.32\textwidth]{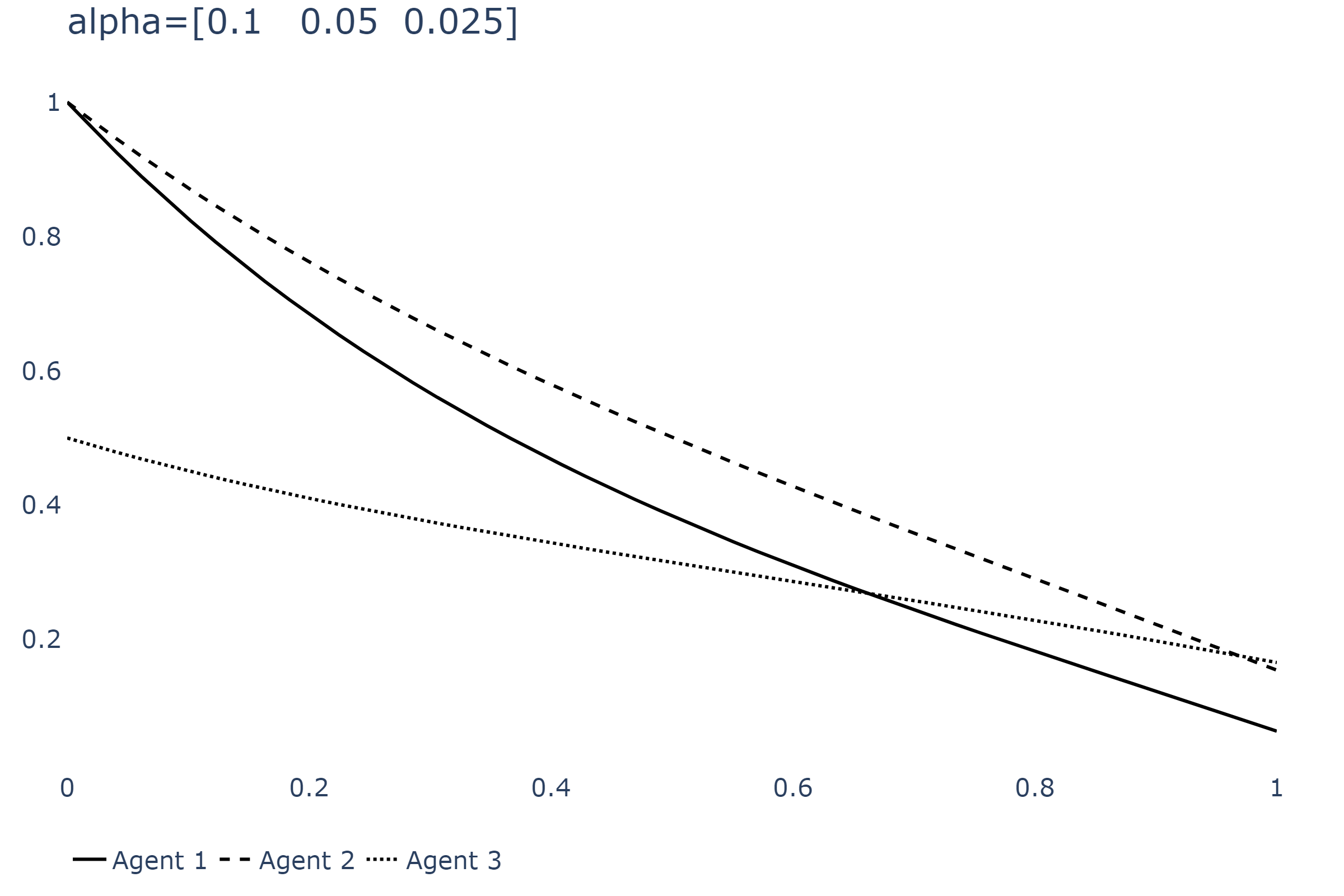}}
            \subfigure{\includegraphics[width=0.32\textwidth]{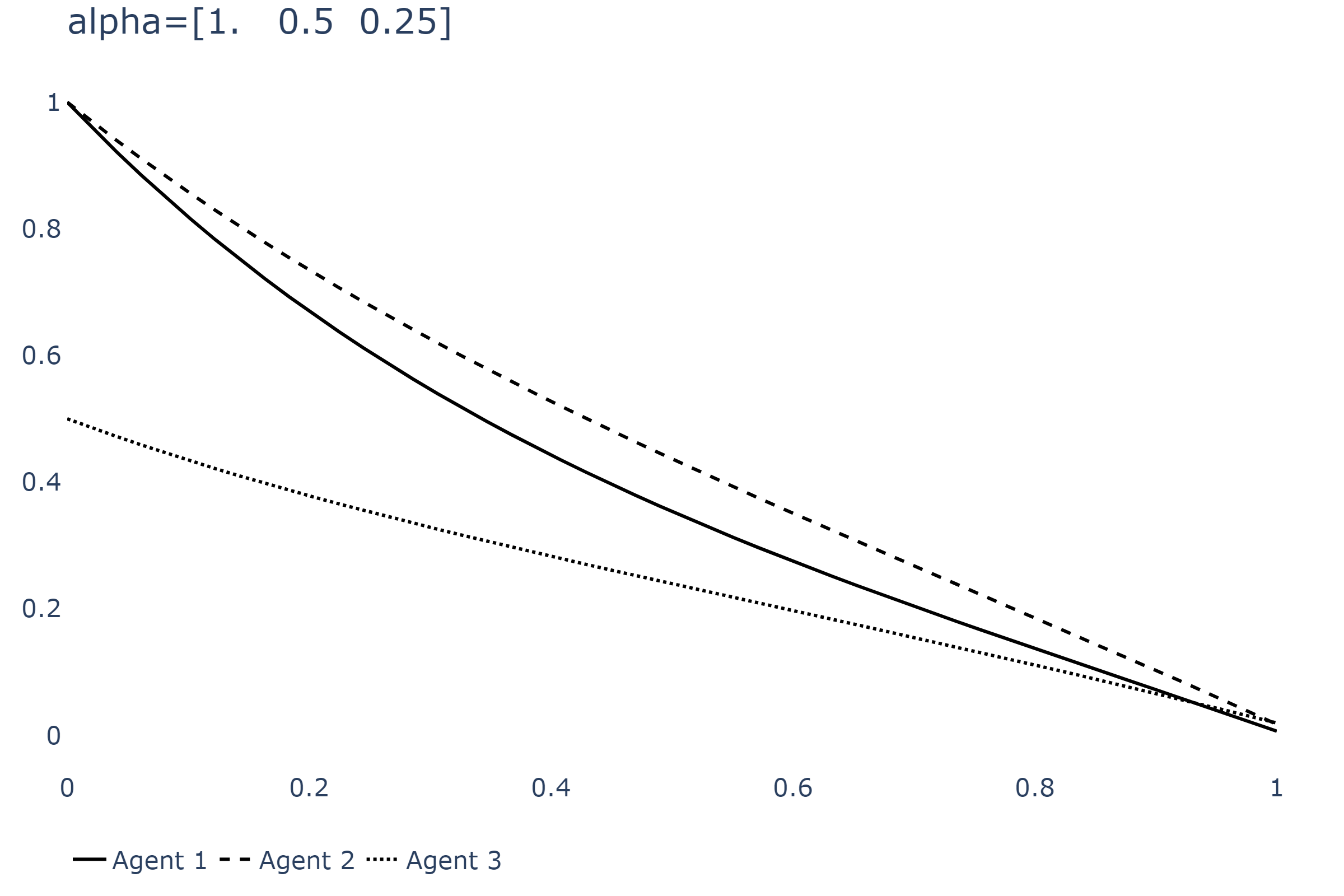}}
            \subfigure{\includegraphics[width=0.32\textwidth]{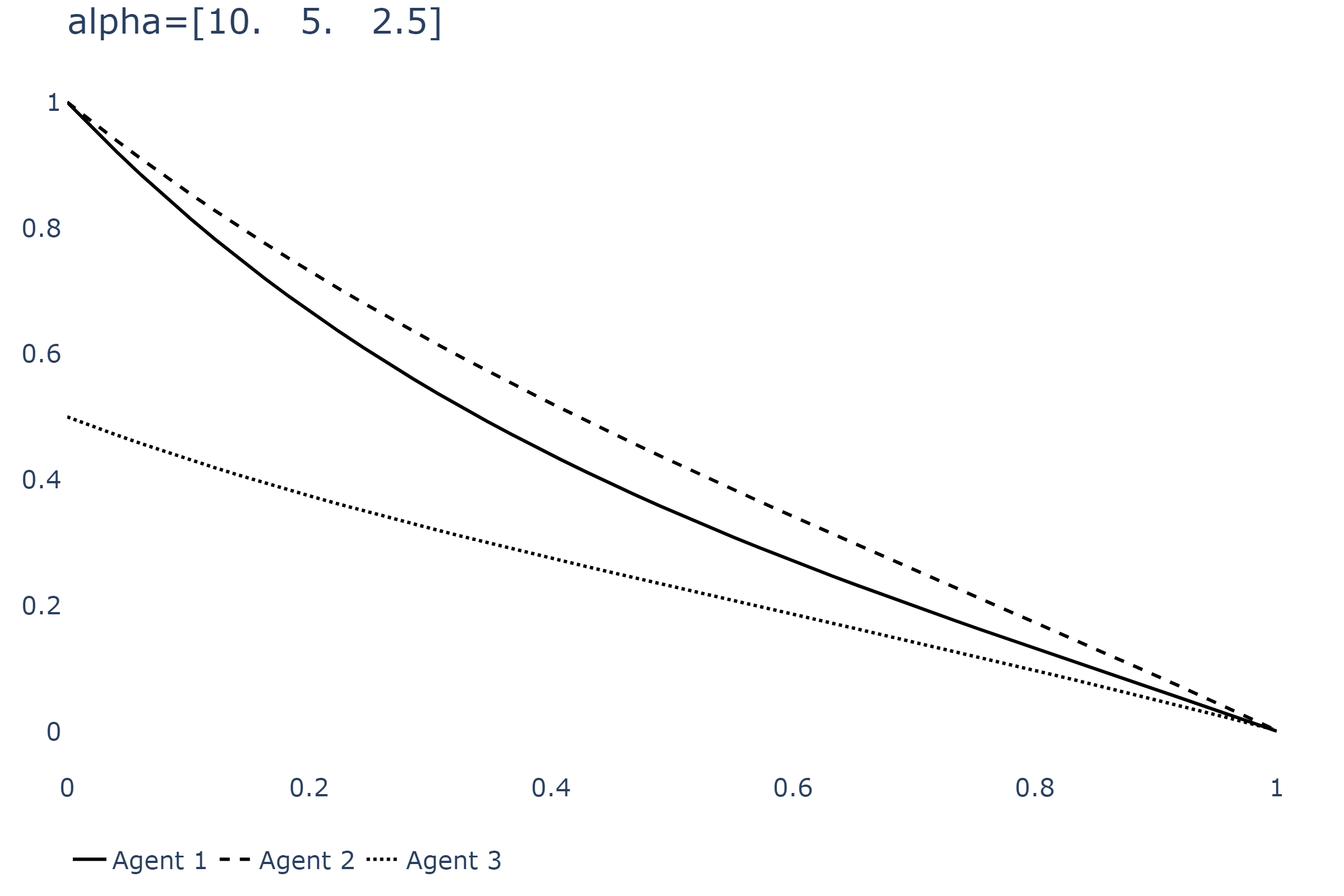}}
            \caption{Plot of the three agents'inventory for different values of risk aversion on terminal value}
            \label{fig:alphadep}
        \end{figure}
    \item Dependence on $\alpha$ different for the first agent
        \begin{figure}[H]
            \centering
            \subfigure{\includegraphics[width=0.32\textwidth]{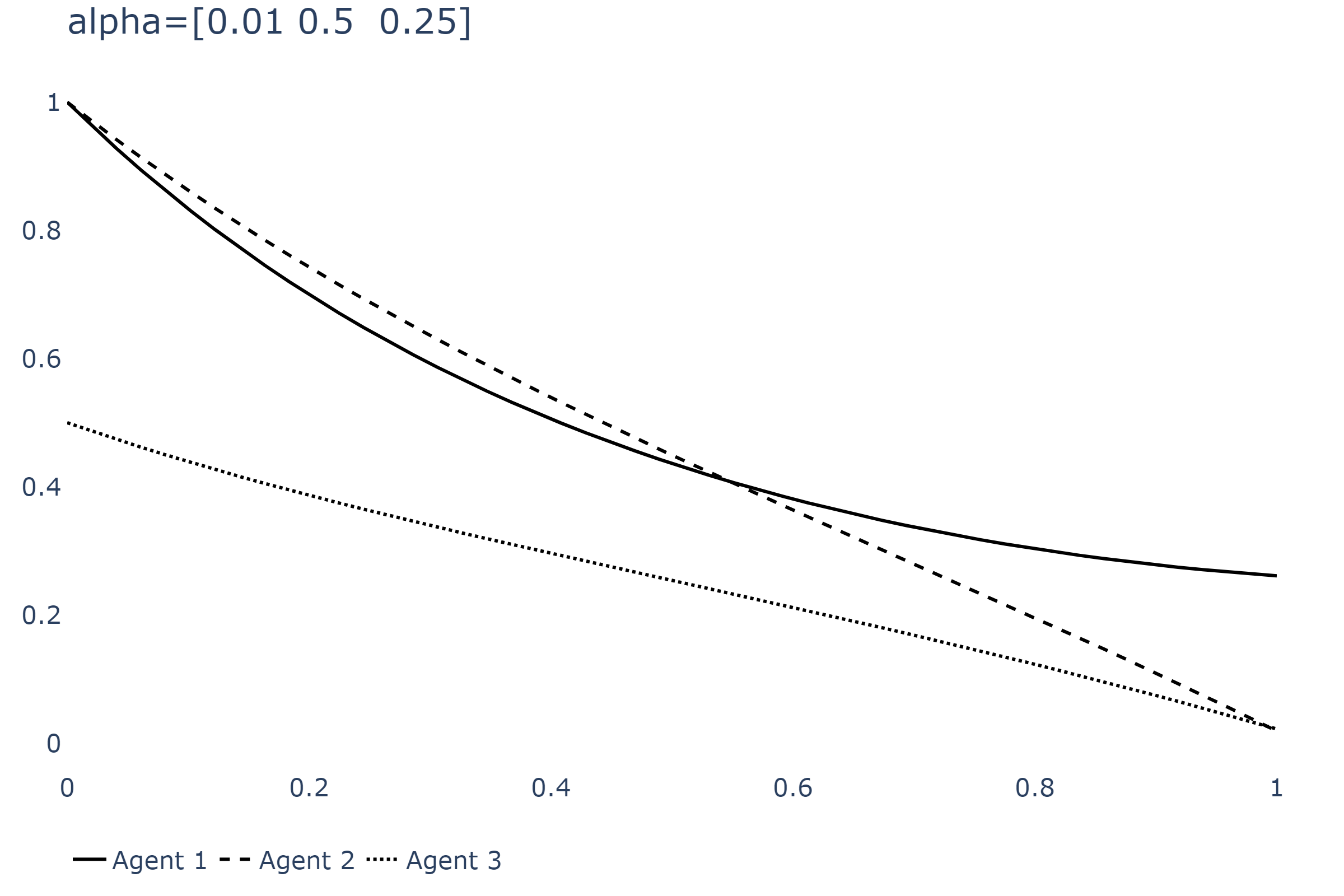}}
            \subfigure{\includegraphics[width=0.32\textwidth]{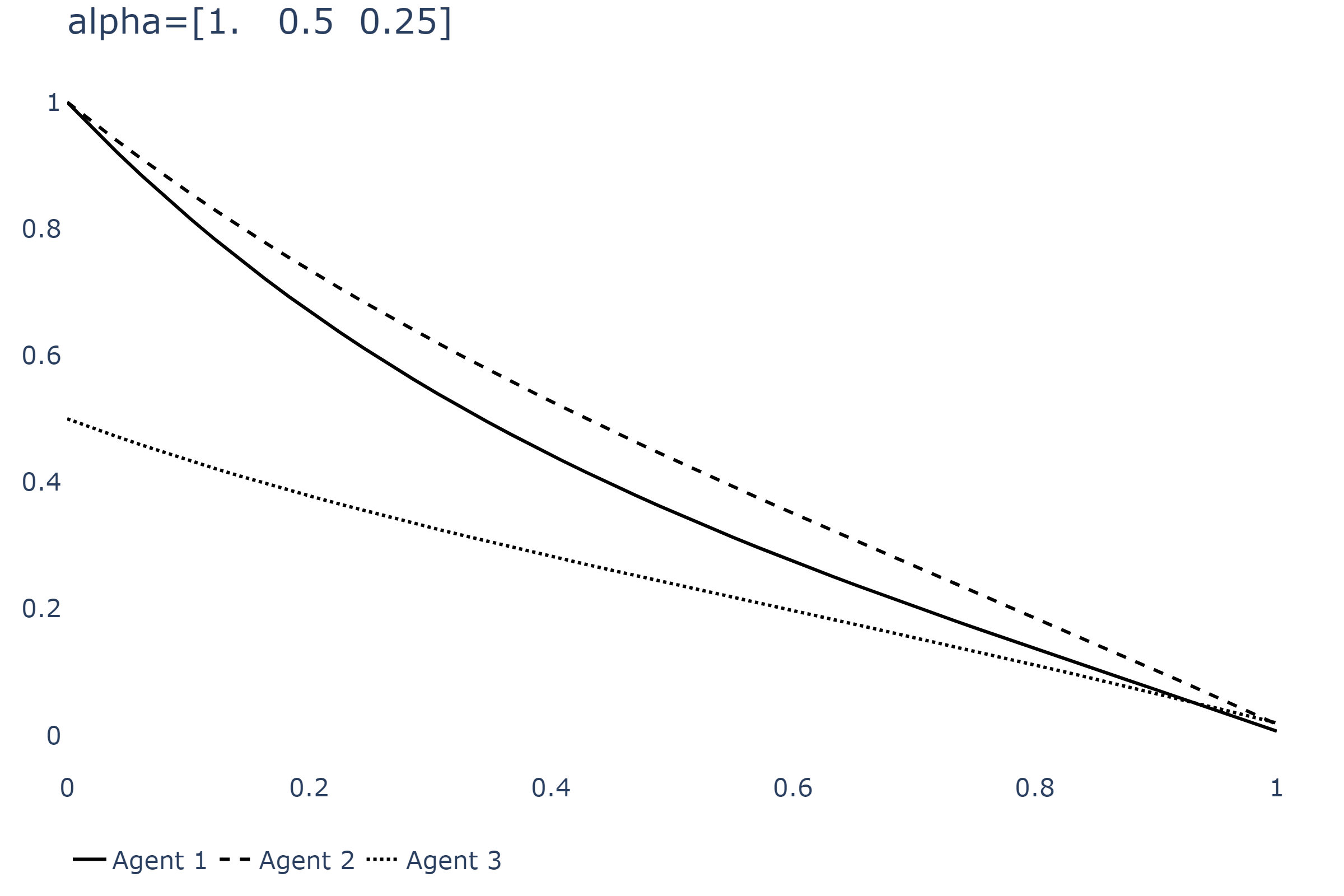}}
            \subfigure{\includegraphics[width=0.32\textwidth]{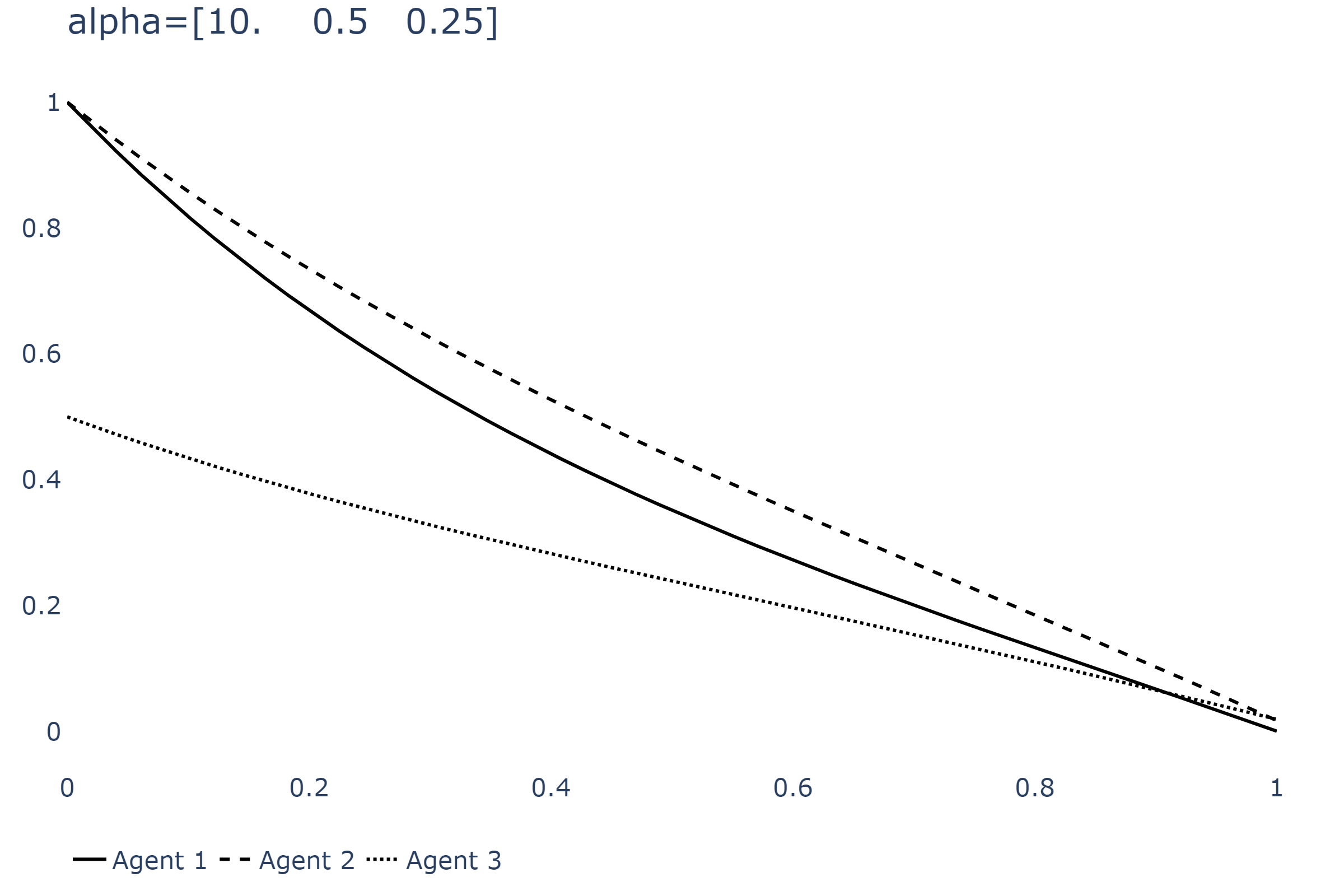}}
            \caption{Plot of the three agents'inventory for different values of risk aversion on terminal value}
            \label{fig:alphadep_unbal}
        \end{figure}
    \item Dependence on $\lambda$ joint magnitude
        \begin{figure}[H]
            \centering
            \subfigure{\includegraphics[width=0.32\textwidth]{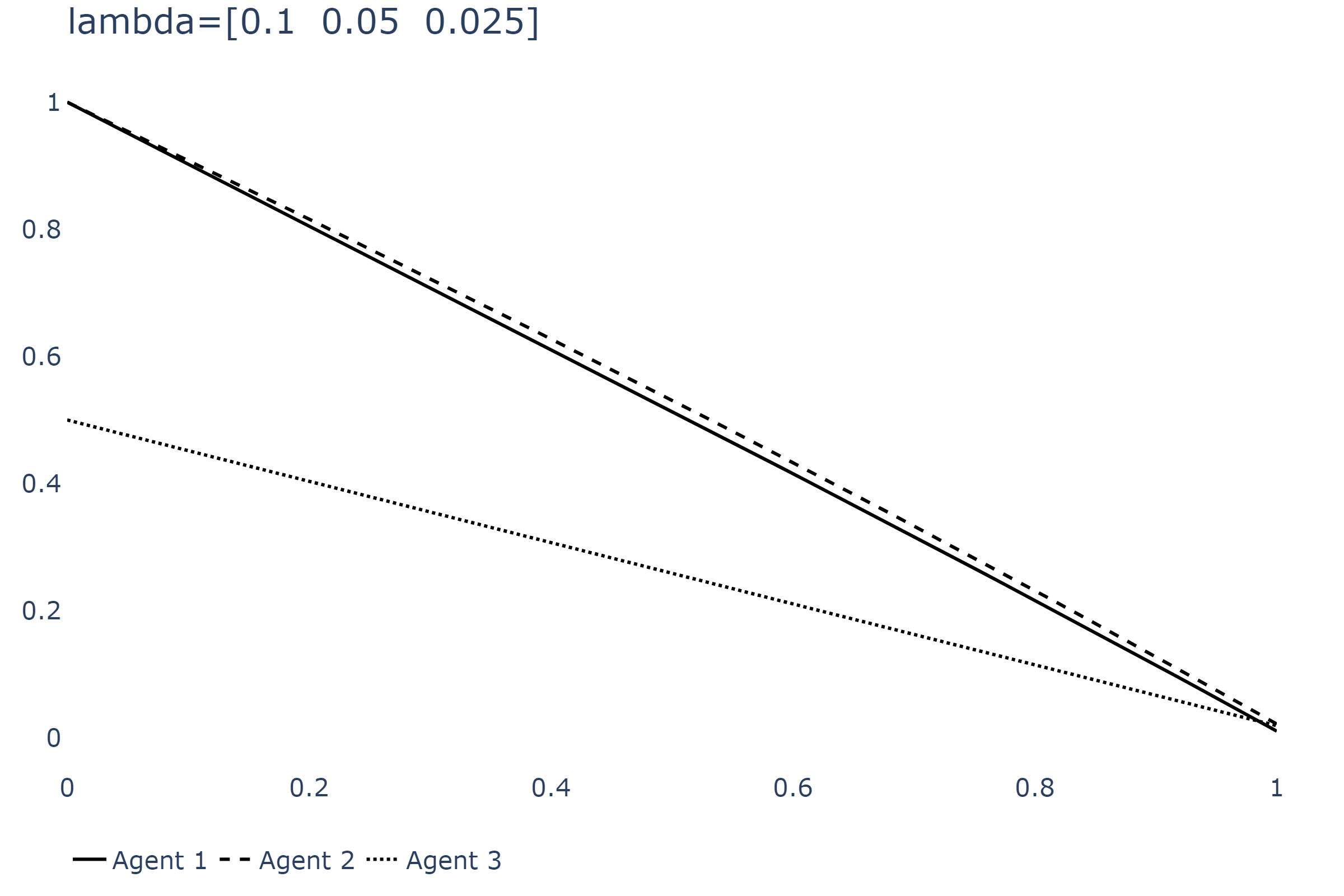}}
            \subfigure{\includegraphics[width=0.32\textwidth]{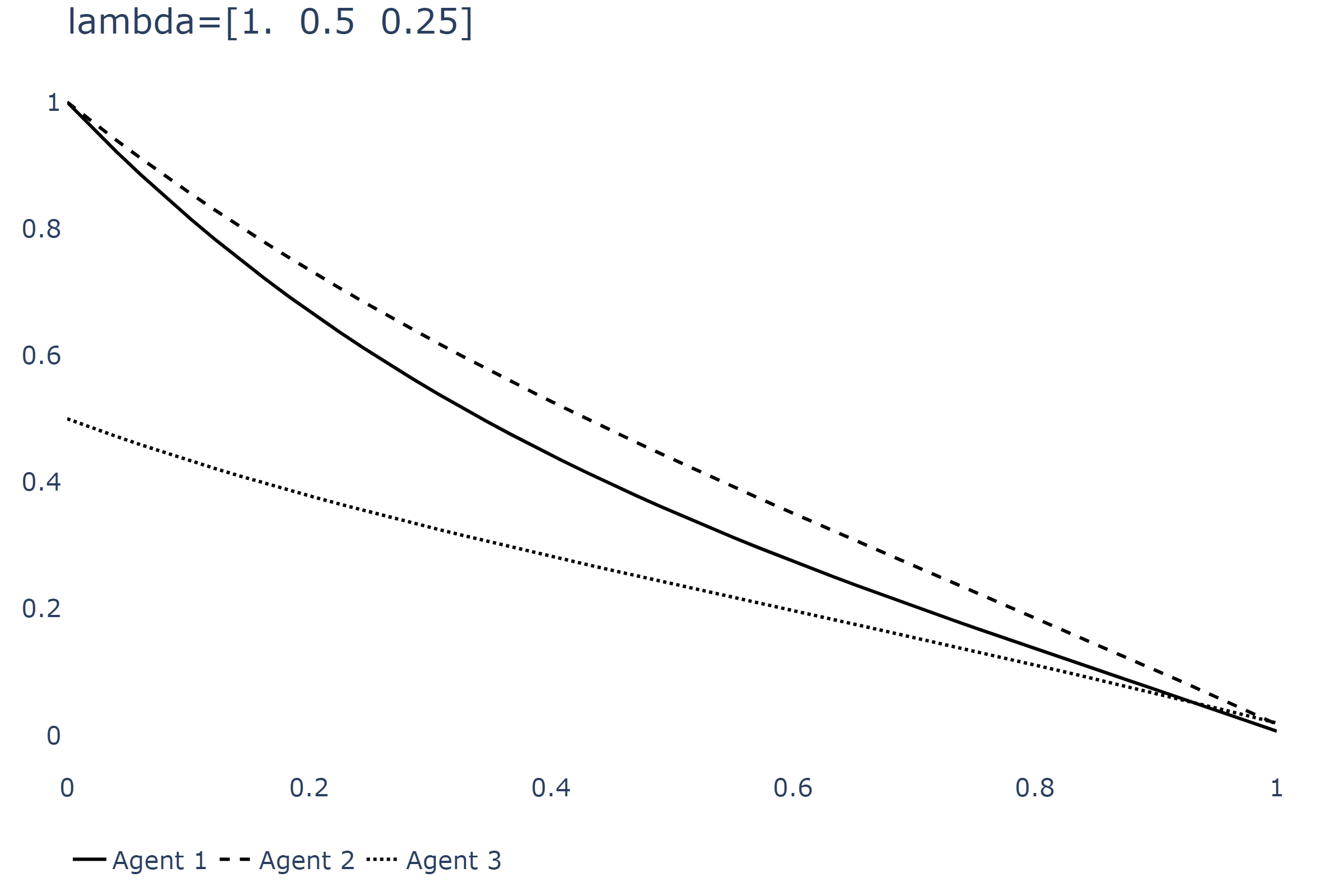}}
            \subfigure{\includegraphics[width=0.32\textwidth]{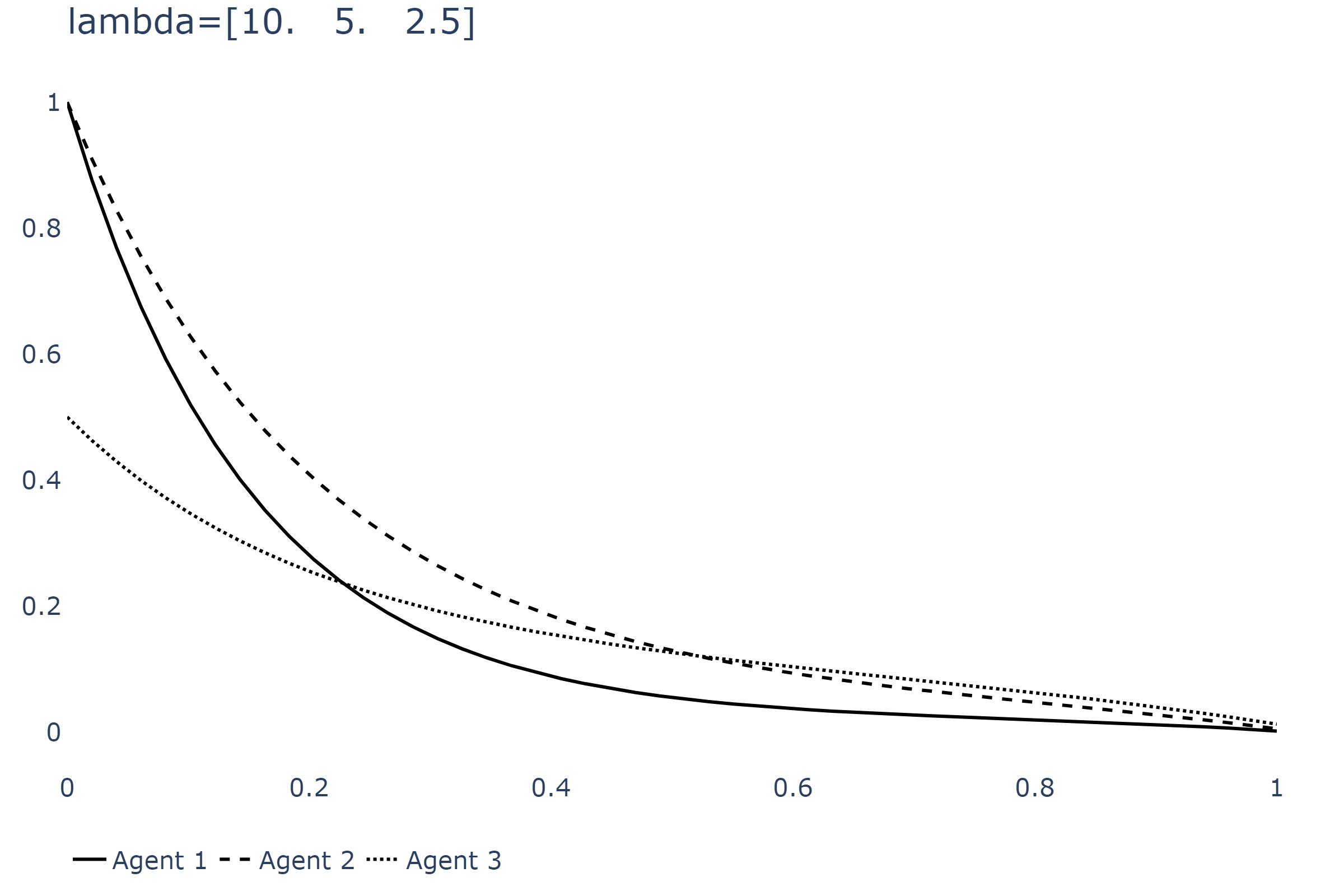}}
            \caption{Plot of the three agents'inventory for different values of risk aversion on continuous trading}
            \label{fig:lambdadep}
        \end{figure}
    \item Dependence on $\lambda$ first agent
        \begin{figure}[H]
            \centering
            \subfigure{\includegraphics[width=0.32\textwidth]{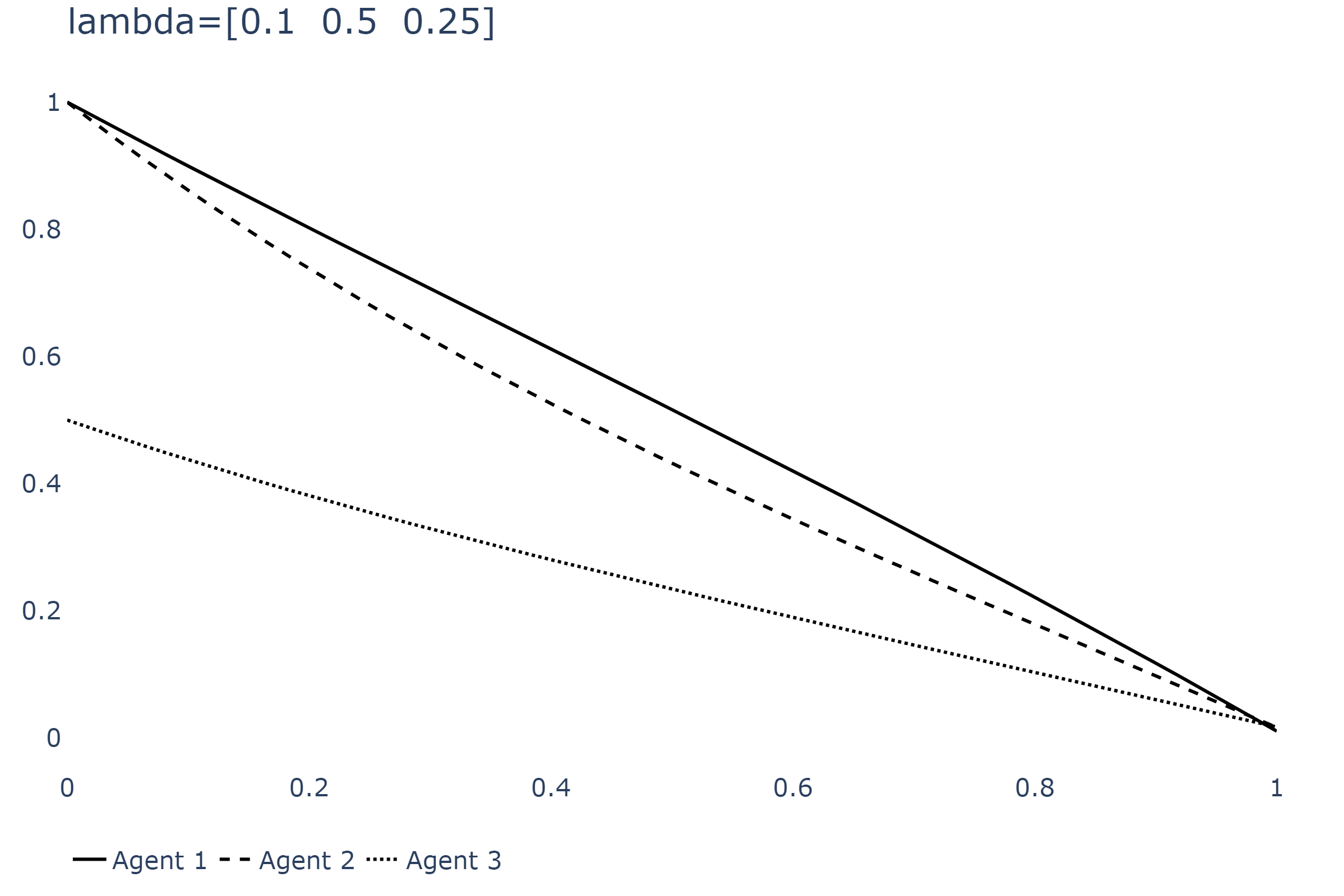}}
            \subfigure{\includegraphics[width=0.32\textwidth]{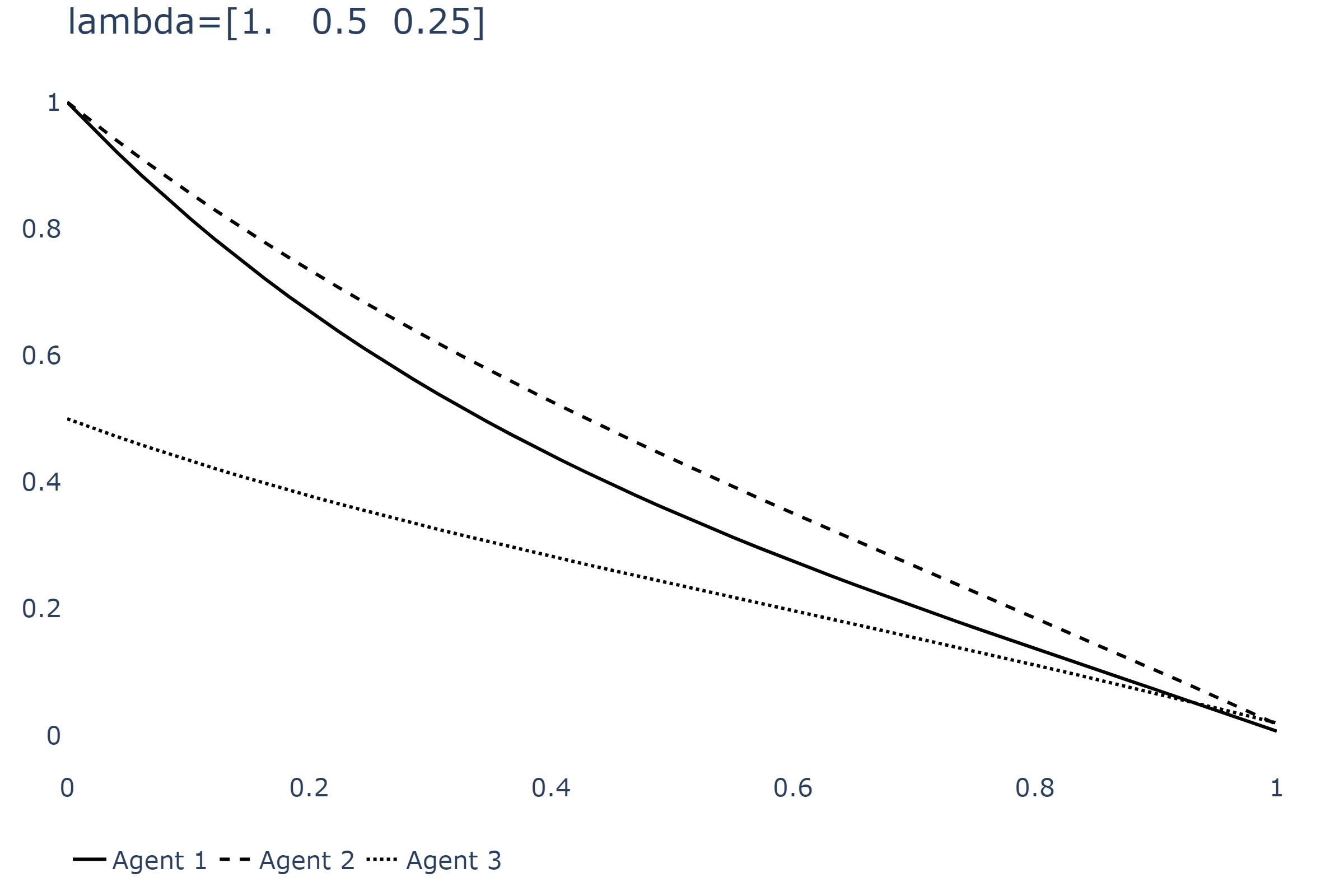}}
            \subfigure{\includegraphics[width=0.32\textwidth]{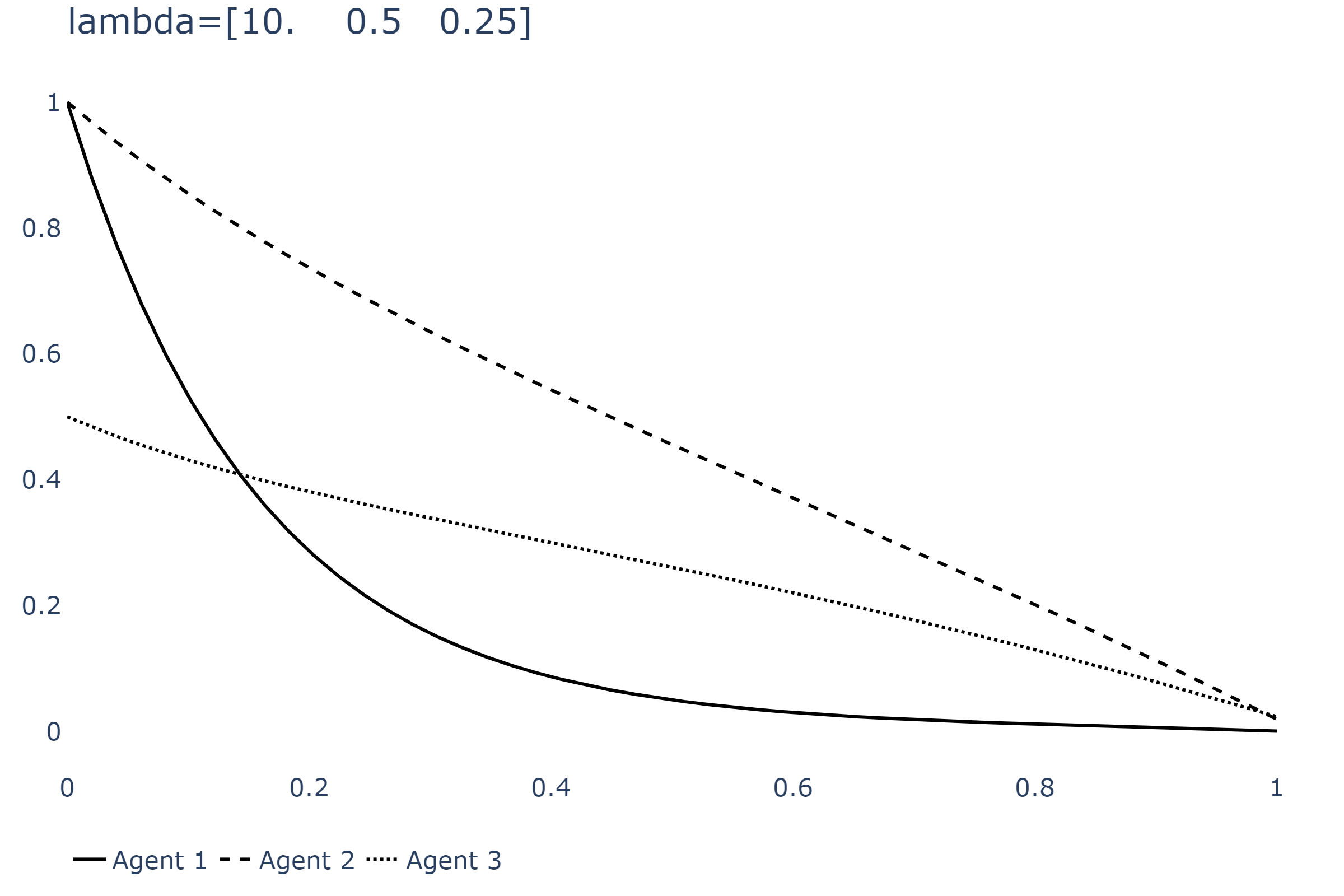}}
            \caption{Plot of the three agents'inventory for different values of risk aversion on continuous trading for one agent}
            \label{fig:lambdadep_unb}
        \end{figure}
    \item Dependence on $Q$ first agent
        \begin{figure}[H]
            \centering
            \subfigure{\includegraphics[width=0.32\textwidth]{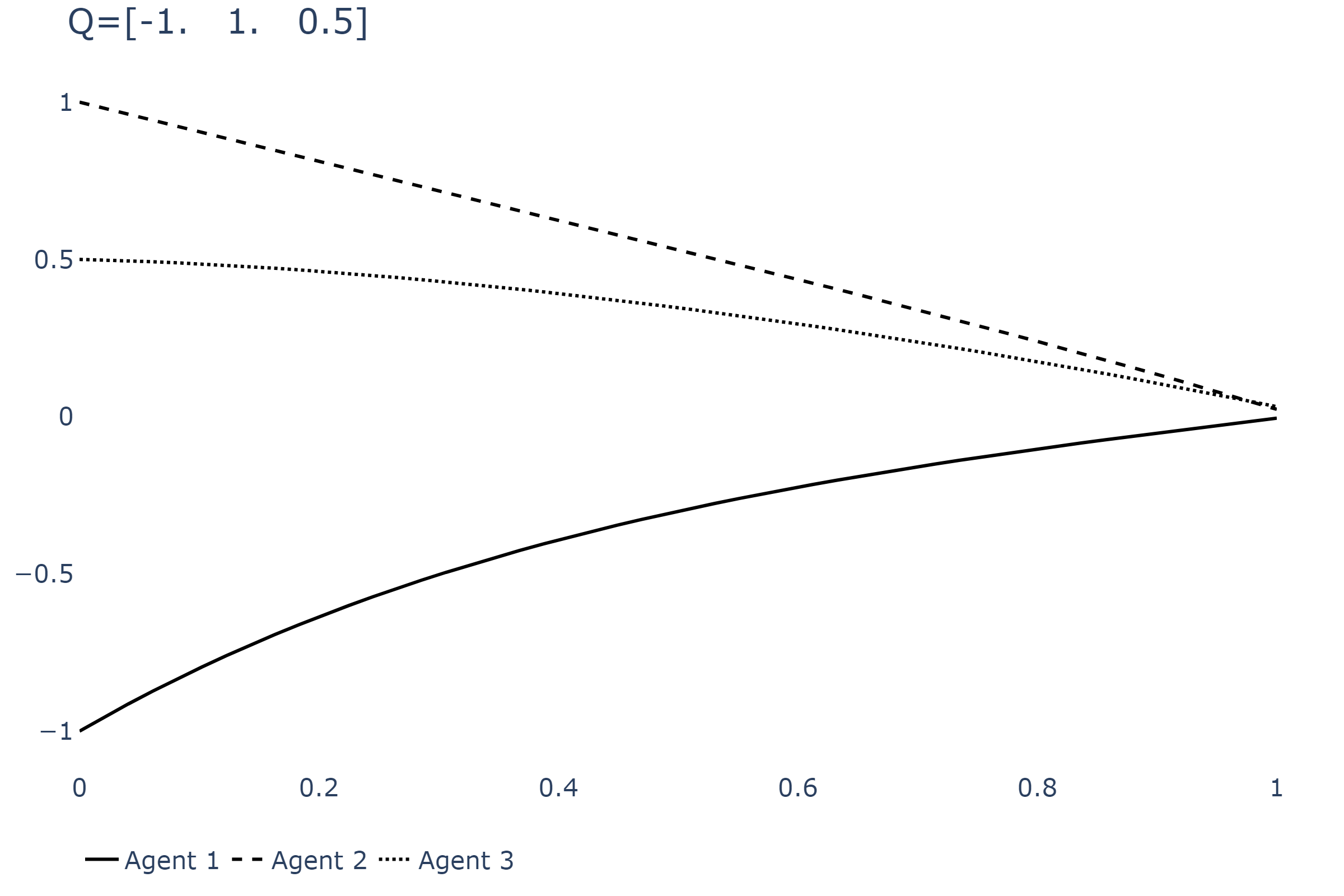}}
            \subfigure{\includegraphics[width=0.32\textwidth]{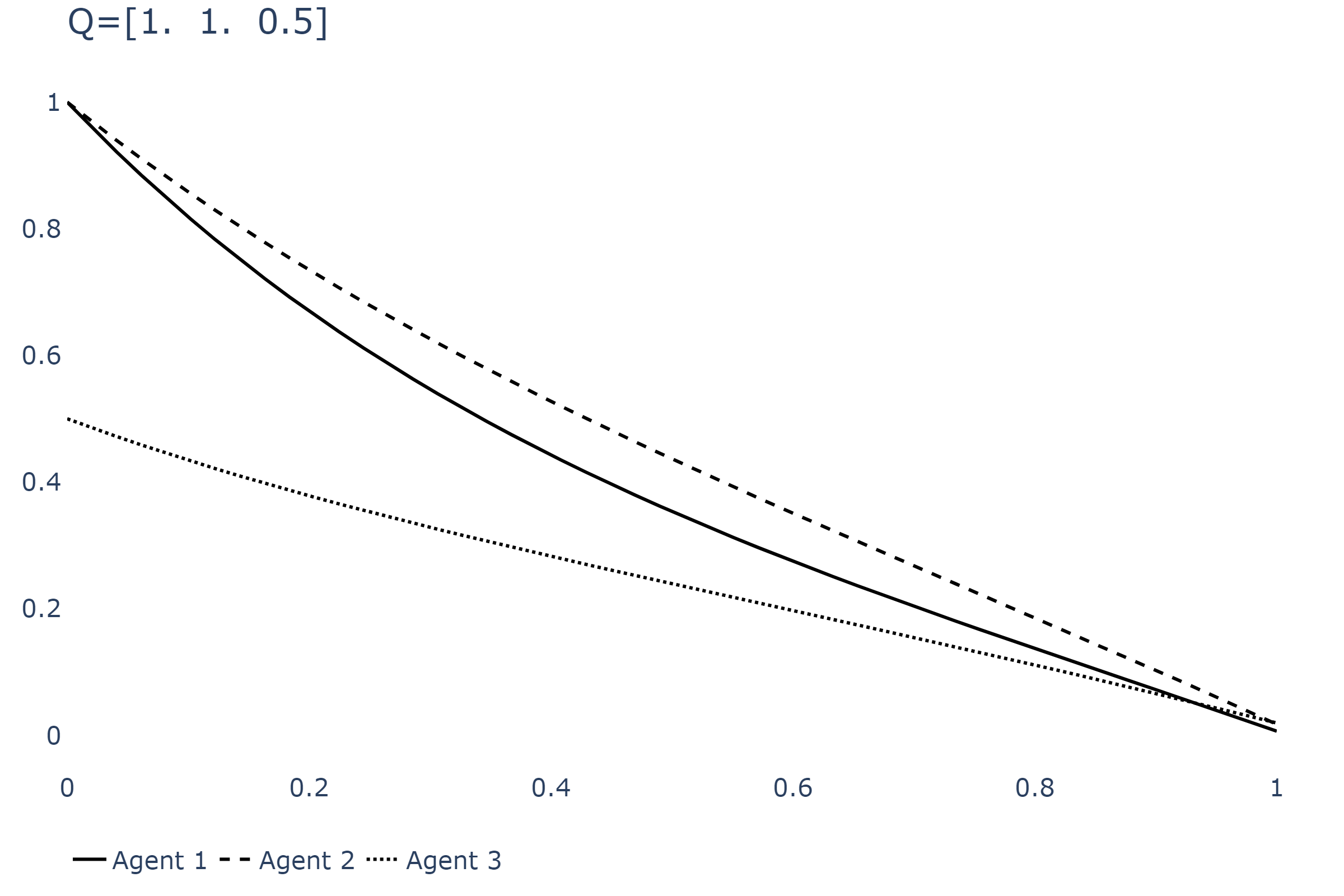}}
            \subfigure{\includegraphics[width=0.32\textwidth]{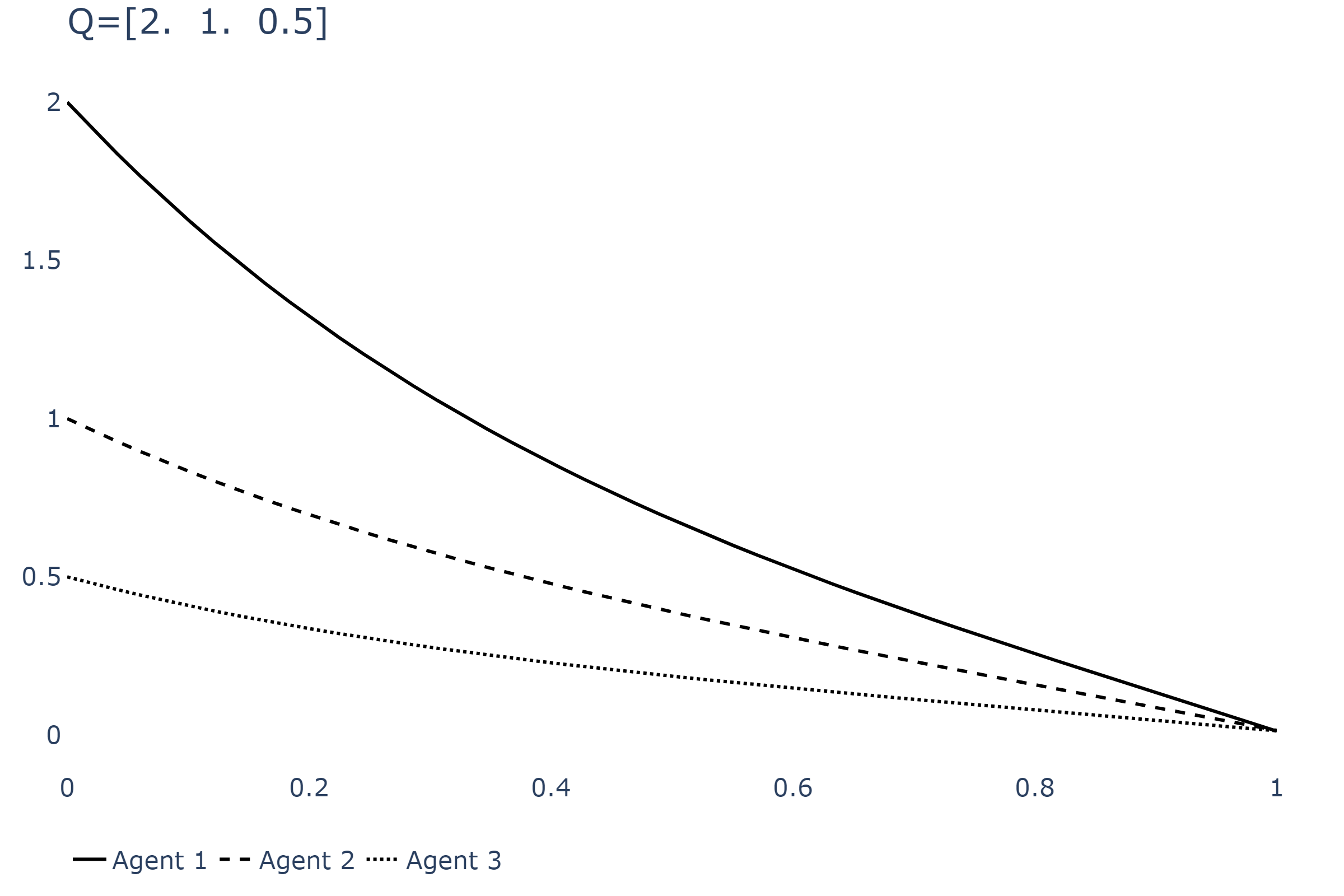}}
            \caption{Plot of the three agents'inventory for different start value of fist agent's inventory.}
            \label{fig:Qdep}
        \end{figure}
    \item Dependence on $Q$ first agent with two arbitrageurs
        \begin{figure}[H]
            \centering
            \subfigure{\includegraphics[width=0.32\textwidth]{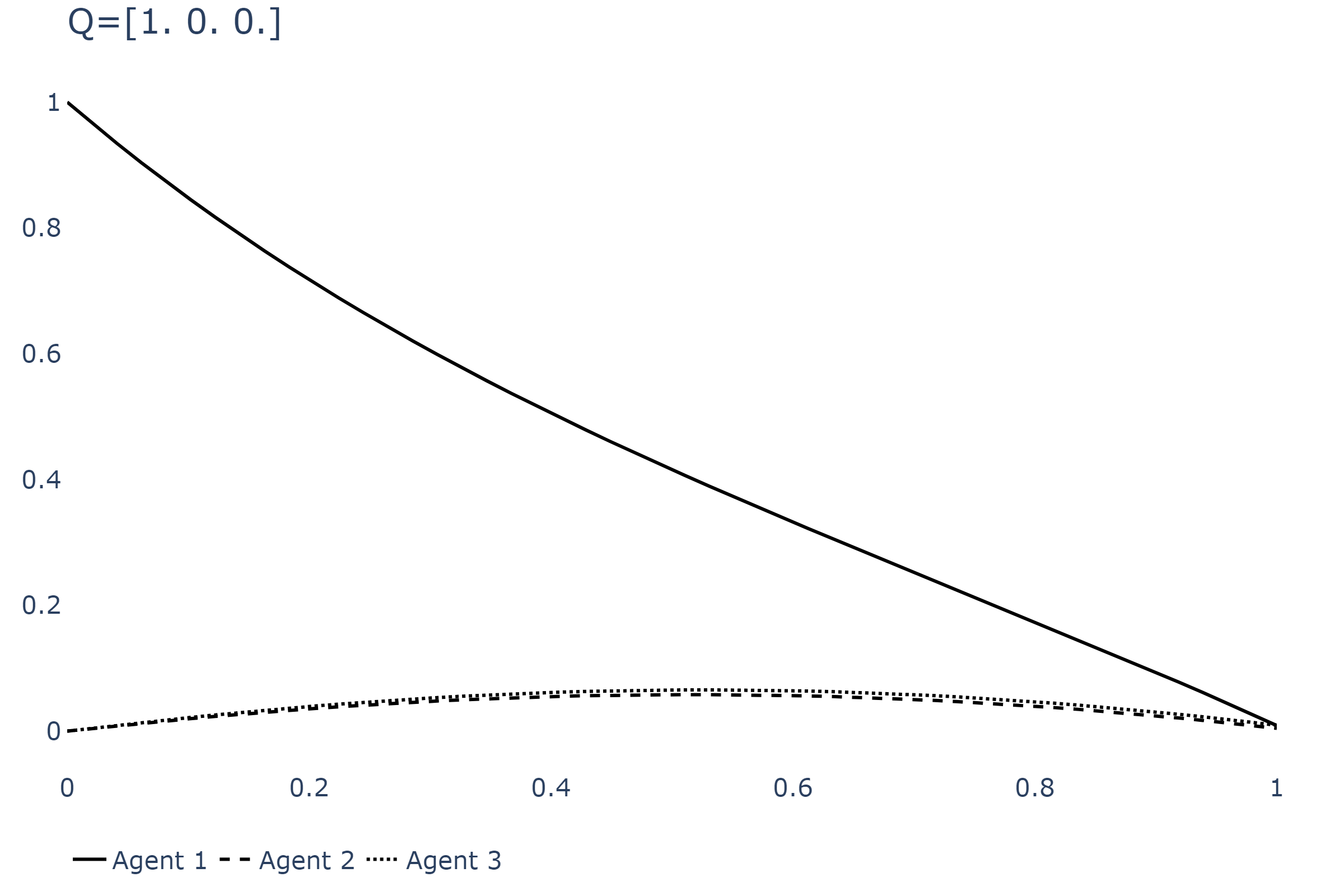}}
            \subfigure{\includegraphics[width=0.32\textwidth]{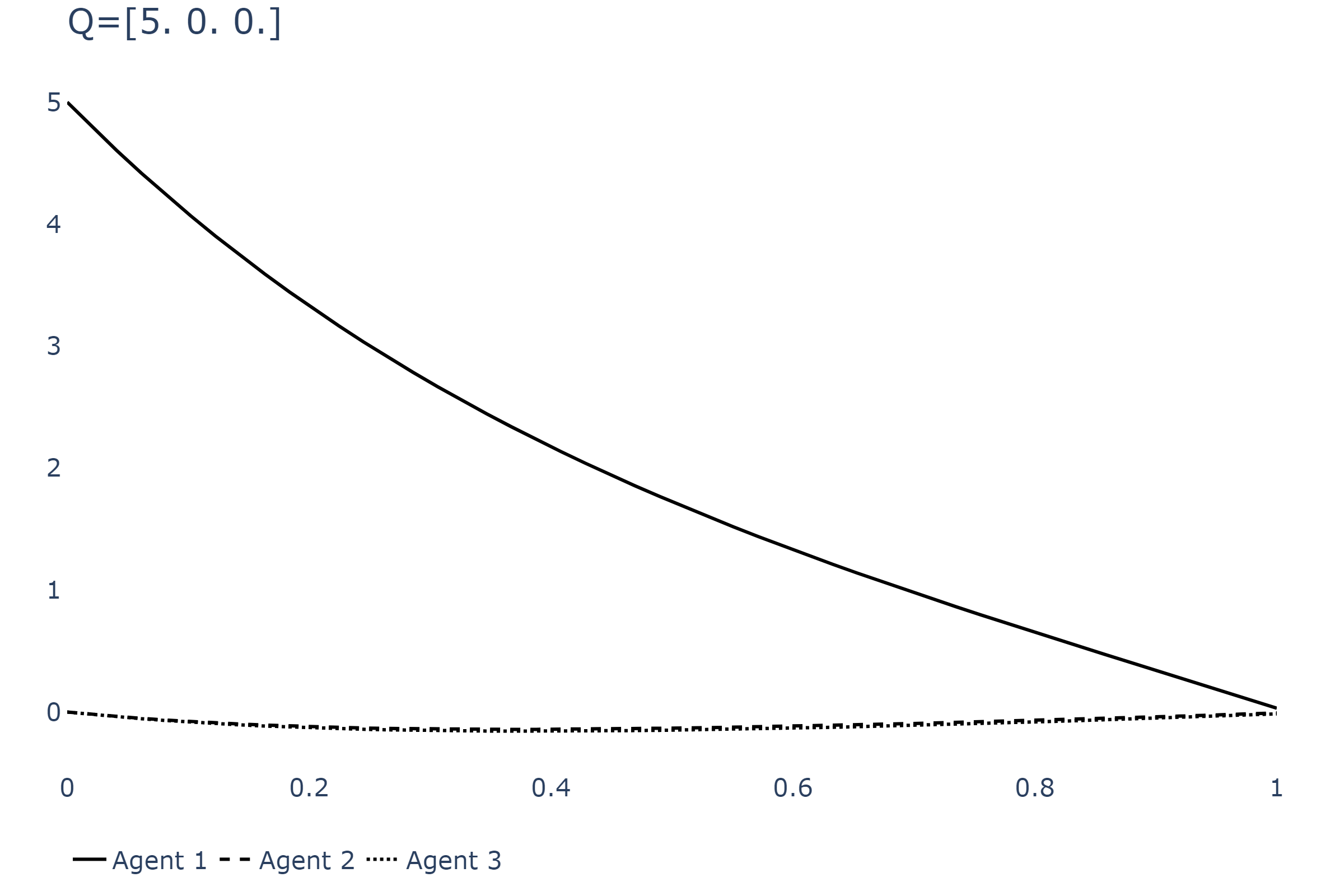}}
            \subfigure{\includegraphics[width=0.32\textwidth]{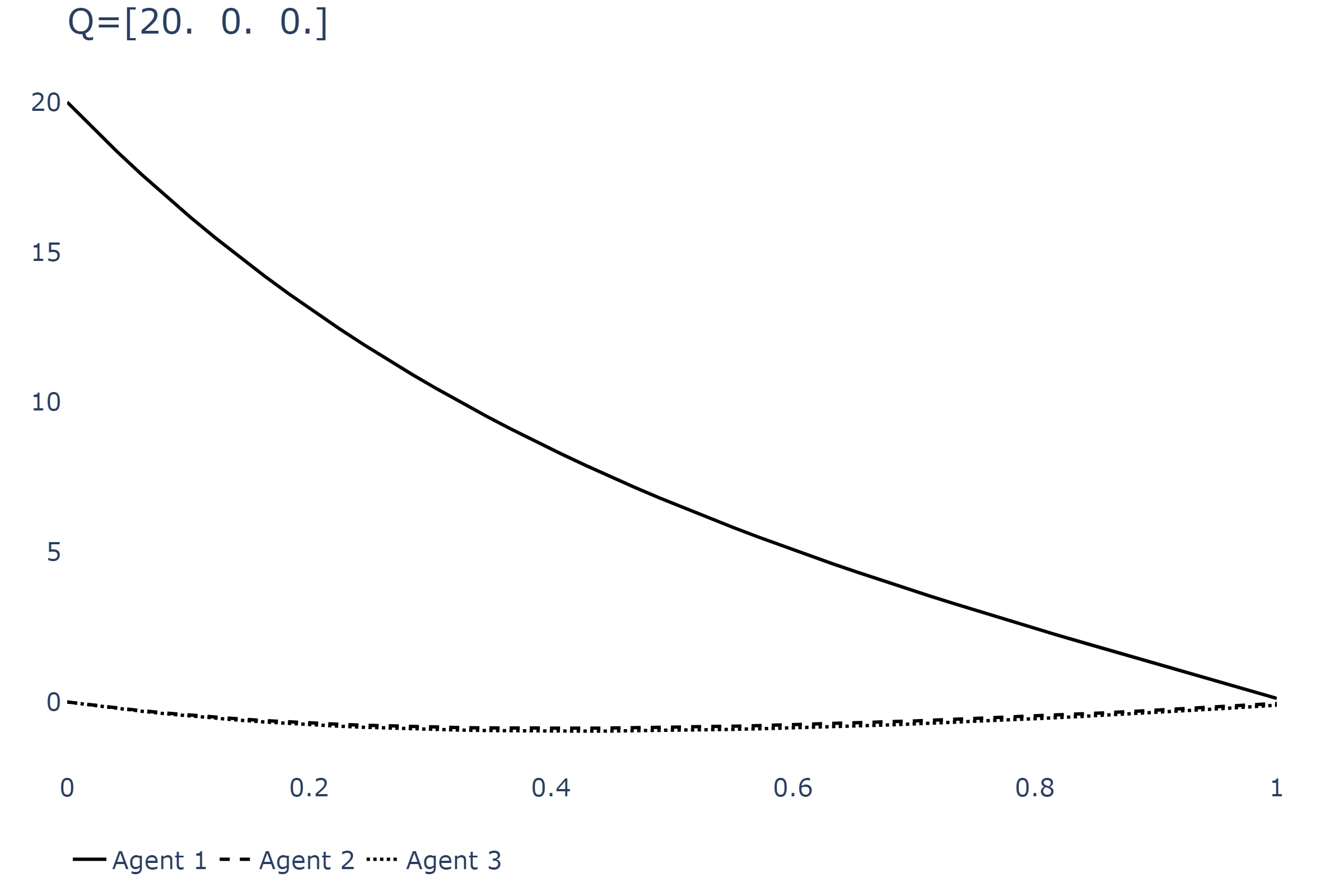}}
            \caption{Plot of the three agents'inventory for different start value of fist agent's inventory with two arbitrageurs.}
            \label{fig:QATdep}
        \end{figure}
        When the initial position of the first agent is small, arbitrageurs tend to first buy and then liquidate to benefit from the future mean return. When the initial position of the first agent is high, arbitrageurs tend to first short sell and then buy to make profit from the price differences. Moreover, The arbitrageurs will not use very aggressive strategies.
\end{itemize}
\subsection{Case for similar agents}
\begin{theorem}\label{SIM}
 Suppose that $\beta_1=\ldots=\beta_n=\beta\geq 0$ and $\lambda_1=\ldots=\lambda_n=\lambda\geq 0$. Then
 \begin{equation}\label{FBSDE1}
\begin{cases}
\displaystyle &\tilde{Q}_t=\sum_{i=1}^{n}Q^i_0+\int_0^t\tilde{q}_sds,\\
\displaystyle &\tilde{q}_t=-\frac{\beta}{b}\tilde{Q}_T+\int_t^T\frac{1}{b}\left(-\lambda\sigma^2_s\tilde{Q}_s+\frac{n\mu_s}{2}+\frac{(n-1)a\tilde{q}_s}{2}\right)ds+\int_t^T\tilde{Z}_sdW_s
\end{cases}
\end{equation}
admits a unique solution $(\tilde{Q},\tilde{q},\tilde{Z})\in\mathcal{S}^2(\mathbb{R})\times\mathcal{S}^2(\mathbb{R})\times\mathcal{H}^2(\mathbb{R}^{d})$ and
\begin{equation}\label{FBSDE2}
\begin{cases}
\displaystyle &Q^{q^i}_t=Q^i_0+\int_0^tq^i_sds, ~~ i=1,\ldots,n\\
\displaystyle &q^i_t=-\frac{\beta}{b}Q^{q^i}_T+\int_t^T\frac{1}{b}\left(-\lambda\sigma^2_sQ^{q^i}_s+\frac{\mu_s}{2}-\frac{aq^i_t}{2}+\frac{a\tilde{q}_s}{2}\right)ds+\int_t^TZ^i_sdW_s, ~~i=1,\ldots,n
\end{cases}
\end{equation}
admits a unique solution $(Q^q,q,Z)\in\mathcal{S}^2(\mathbb{R}^n)\times\mathcal{S}^2(\mathbb{R}^n)\times\mathcal{H}^2(\mathbb{R}^{n\times d})$. In addition, it holds that
\begin{equation*}
\begin{cases}
&\tilde{Q}=\sum_{i=1}^{n}Q^{q^i},\\
&\tilde{q}=\sum_{i=1}^{n}q^i,\\
&\tilde{Z}=\sum_{i=1}^{n}Z^i.
\end{cases}
\end{equation*}
Moreover, $(Q^q,q,Z)$ is the unique solution of FBSDE \eqref{FBSDE}.
\end{theorem}
\begin{proof}
We will divide the proof into several steps.

\emph{Step 1:} Denoting $M=e^{\frac{(n-1)aT}{2b}}\frac{\beta}{b}+e^{\frac{(n-1)aT}{2b}}\frac{\lambda\|\sigma\|_{\infty}^2}{2b}T
$, it follows from Pardoux and Peng \cite{PP} that BSDE
\begin{equation*}
P_t=-\frac{\beta}{b}+\int_t^T\left(-\frac{\lambda \sigma^2_s}{b}+\frac{(n-1)a}{2b}P_s+\left(\left(-M\right)\vee P_s\wedge M\right)^2\right)ds-\int_t^T\Lambda_sdW_s
\end{equation*}
admits a unique solution $(P,\Lambda)\in\mathcal{S}^2(\mathbb{R})\times\mathcal{H}^2(\mathbb{R}^d)$. Moreover, we have the following a priori estimate for $P$. Since
\begin{equation*}
e^{\frac{(n-1)at}{2b}}P_t=-e^{\frac{(n-1)aT}{2b}}\frac{\beta}{b}+\int_t^T\left(-e^{\frac{(n-1)as}{2b}}\frac{\lambda\sigma^2_s}{2b}+e^{\frac{(n-1)as}{2b}}\left(\left(-M\right)\vee P_s\wedge M\right)^2\right)ds-\int_t^Te^{\frac{(n-1)as}{2b}}\Lambda_sdW_s
\end{equation*}
it holds that
\begin{align*}
e^{\frac{(n-1)at}{2b}}P_t & \geq E\left[-e^{\frac{(n-1)aT}{2b}}\frac{\beta}{b}+\int_t^T\left(-e^{\frac{(n-1)as}{2b}}\frac{\lambda\sigma^2_s}{2b}\right)ds\Bigg|\mathcal{F}_t\right]\\
&\geq -e^{\frac{(n-1)aT}{2b}}\frac{\beta}{b}-e^{\frac{(n-1)aT}{2b}}\frac{\lambda\|\sigma\|_{\infty}^2}{2b}\left(T-t\right)
\end{align*}
Hence
\begin{equation*}
P_t\geq -e^{\frac{(n-1)aT}{2b}}\frac{\beta}{b}-e^{\frac{(n-1)aT}{2b}}\frac{\lambda\|\sigma\|_{\infty}^2}{2b}T
\end{equation*}
Meanwhile by denoting $\xi_t=\left(-M\right)\vee P_t\wedge M$, it holds that
\begin{align*}
e^{\int_0^t\left(\xi_s+\frac{(n-1)a}{2b}\right)ds}P_t&=-e^{\int_0^T\left(\xi_s+\frac{(n-1)a}{2b}\right)ds}\frac{\beta}{b}+\int_t^Te^{\int_0^s\left(\xi_u+\frac{(n-1)a}{2b}\right)du}\left(-\frac{\lambda\sigma^2_s}{b}-\xi_sP_s+\xi^2_s\right)ds\\
&\quad\quad\quad-\int_t^Te^{\int_0^s\left(\xi_u+\frac{(n-1)a}{2b}\right)du}\Lambda_sdW_s\\
&\leq -e^{\int_0^T\left(\xi_s+\frac{(n-1)a}{2b}\right)ds}\frac{\beta}{b}-\int_t^Te^{\int_0^s\left(\xi_u+\frac{(n-1)a}{2b}\right)du}\Lambda_sdW_s.
\end{align*}
Therefore, we have
\begin{align*}
P_t&\leq E\left[-e^{\int_t^T\left(\xi_s+\frac{(n-1)a}{2b}\right)ds}\frac{\beta}{b}\Bigg|\mathcal{F}_t\right]\\
&\leq -\frac{\beta}{b}e^{-M(T-t)}.
\end{align*}
Hence, $(P,\Lambda)\in\mathcal{S}^{\infty}(\mathbb{R})\times\mathcal{H}^2(\mathbb{R}^d)$ and satisfies
\begin{equation*}
P_t=-\frac{\beta}{b}+\int_t^T\left(-\frac{\lambda \sigma^2_s}{b}+\frac{(n-1)a}{2b}P_s+P^2_s\right)ds-\int_t^T\Lambda_sdW_s
\end{equation*}
On the other hand, if
\begin{equation*}
P_t=-\frac{\beta}{b}+\int_t^T\left(-\frac{\lambda \sigma^2_s}{b}+\frac{(n-1)a}{2b}P_s+P^2_s\right)ds-\int_t^T\Lambda_sdW_s
\end{equation*}
admits a solution $(P,\Lambda)\in\mathcal{S}^{\infty}(\mathbb{R})\times\mathcal{H}^2(\mathbb{R}^d)$, we have
\begin{equation*}
e^{\frac{(n-1)at}{2b}}P_t=-e^{\frac{(n-1)aT}{2b}}\frac{\beta}{b}+\int_t^T\left(-e^{\frac{(n-1)as}{2b}}\frac{\lambda\sigma^2_s}{2b}+e^{\frac{(n-1)as}{2b}} P^2_s\right)ds-\int_t^Te^{\frac{(n-1)as}{2b}}\Lambda_sdW_s
\end{equation*}
and
\begin{equation*}
e^{\int_0^t\left(\frac{(n-1)a}{2b}+P_s\right)ds}P_t=-e^{\int_0^T\left(\frac{(n-1)a}{2b}+P_s\right)ds}\frac{\beta}{2b}+\int_t^T\left(-e^{\int_0^s\left(\frac{(n-1)a}{2b}+P_u\right)du}\frac{\lambda\sigma^2_s}{2b}\right)ds-\int_t^Te^{\int_0^s\left(\frac{(n-1)a}{2b}+P_u\right)du}\Lambda_sdW_s
\end{equation*}
Therefore, we have
\begin{align*}
P_t & \geq E\left[-e^{\frac{(n-1)a(T-t)}{2b}}\frac{\beta}{b}+\int_t^T\left(-e^{\frac{(n-1)a(s-t)}{2b}}\frac{\lambda\sigma^2_s}{2b}\right)ds\Bigg|\mathcal{F}_t\right]\\
&\geq -e^{\frac{(n-1)aT}{2b}}\frac{\beta}{b}-e^{\frac{(n-1)aT}{2b}}\frac{\lambda\|\sigma\|_{\infty}^2}{2b}T
\end{align*}
and
\begin{align*}
P_t&\leq E\left[-e^{\int_t^T\left(P_s+\frac{(n-1)a}{2b}\right)ds}\frac{\beta}{b}\Bigg|\mathcal{F}_t\right]\\
&\leq -\frac{\beta}{b}e^{-M(T-t)}.
\end{align*}
Hence, $(P,\Lambda)$ satisfies
\begin{equation*}
P_t=-\frac{\beta}{b}+\int_t^T\left(-\frac{\lambda \sigma^2_s}{b}+\frac{(n-1)a}{2b}P_s+\left(\left(-M\right)\vee P_s\wedge M\right)^2\right)ds-\int_t^T\Lambda_sdW_s
\end{equation*}
Again, it follows from Pardoux and Peng \cite{PP} that
\begin{equation*}
p_t=\int_t^T\left(\left(\frac{(n-1)a}{2b}+P_s\right)p_s+\frac{n\mu_s}{2b}\right)ds-\int_t^T\eta_sdW_s
\end{equation*}
admits a unique solution $(p,\eta)\in\mathcal{S}^2(\mathbb{R})\times\mathcal{H}^2(\mathbb{R}^d)$. Moreover, one could easily check that $p\in\mathcal{S}^{\infty}(\mathbb{R})$, $\eta\in \text{\rm BMO}(\mathbb{R}^d)$ and $\Lambda\in \text{\rm BMO}(\mathbb{R}^d)$. Hence, from standard theory of SDEs, SDE \eqref{SDE} admits a unique strong solution $\tilde{Q}\in\mathcal{S}^{\infty}(\mathbb{R})$. Therefore, according to Proposition \ref{Ric}, FBSDE \eqref{FBSDE1} admits a solution $(\tilde{Q},\tilde{q},\tilde{Z})\in\mathcal{S}^2(\mathbb{R})\times\mathcal{S}^2(\mathbb{R})\times\mathcal{H}^2(\mathbb{R}^{d})$.

We now prove the uniqueness. Suppose that FBSDE \eqref{FBSDE1} admits another solution $(\bar{Q},\bar{q},\bar{Z})\in\mathcal{S}^2(\mathbb{R})\times\mathcal{S}^2(\mathbb{R})\times\mathcal{H}^2(\mathbb{R}^{d})$. Then, we have
 \begin{equation*}
\begin{cases}
\displaystyle &\tilde{Q}_t-\bar{Q}_t=\int_0^t\left(\tilde{q}_s-\bar{q}_s\right)ds,\\
\displaystyle &\tilde{q}_t-\bar{q}_t=-\frac{\beta}{b}\left(\tilde{Q}_T-\bar{Q}_T\right)+\int_t^T\frac{1}{b}\left(-\lambda\sigma^2_s\left(\tilde{Q}_s-\bar{Q}_s\right)+\frac{(n-1)a}{2}\left(\tilde{q}_s-\bar{q}_s\right)\right)ds+\int_t^T\left(\tilde{Z}_s-\bar{Z}_s\right)dW_s
\end{cases}
\end{equation*}
Therefore, it holds that
\begin{align*}
\left(\tilde{q}_t-\bar{q}_t\right)\left(\tilde{Q}_t-\bar{Q}_t\right)&=-\frac{\beta}{b}\left(\tilde{Q}_T-\bar{Q}_T\right)^2+\int_t^T\frac{1}{b}\left(-\lambda\sigma^2_s\left(\tilde{Q}_s-\bar{Q}_s\right)^2+\frac{(n-1)a}{2}\left(\tilde{q}_s-\bar{q}_s\right)\left(\tilde{Q}_t-\bar{Q}_t\right)\right)ds\\
&\quad\quad -\int_t^T\left(\tilde{q}_s-\bar{q}_s\right)^2ds+\int_t^T\left(\tilde{Q}_t-\bar{Q}_t\right)\left(\tilde{Z}_s-\bar{Z}_s\right)dW_s
\end{align*}
Hence, we have
\begin{align*}
e^{\frac{(n-1)at}{2}}\left(\tilde{q}_t-\bar{q}_t\right)\left(\tilde{Q}_t-\bar{Q}_t\right)&=-\frac{\beta}{b}e^{\frac{(n-1)aT}{2}}\left(\tilde{Q}_T-\bar{Q}_T\right)^2-\int_t^T\frac{\lambda\sigma^2_s}{b}e^{\frac{(n-1)as}{2}}\left(\tilde{Q}_s-\bar{Q}_s\right)^2ds\\
&\quad\quad -\int_t^Te^{\frac{(n-1)as}{2}}\left(\tilde{q}_s-\bar{q}_s\right)^2ds+\int_t^Te^{\frac{(n-1)as}{2}}\left(\tilde{Q}_t-\bar{Q}_t\right)\left(\tilde{Z}_s-\bar{Z}_s\right)dW_s
\end{align*}
Thus, it holds that
\begin{align*}
0=E\left[-\frac{\beta}{b}e^{\frac{(n-1)aT}{2}}\left(\tilde{Q}_T-\bar{Q}_T\right)^2-\int_0^T\frac{\lambda\sigma^2_s}{b}e^{\frac{(n-1)as}{2}}\left(\tilde{Q}_s-\bar{Q}_s\right)^2ds -\int_0^Te^{\frac{(n-1)as}{2}}\left(\tilde{q}_s-\bar{q}_s\right)^2ds\right]\leq 0
\end{align*}
which implies uniqueness.

\emph{Step 2:}
 Noting that $\tilde{q}\in\mathcal{S}^{\infty}(\mathbb{R})$, following from a similar technique as in \emph{Step 1}, FBSDE \eqref{FBSDE2} admits a unique solution $(Q^q,q,Z)\in\mathcal{S}^2(\mathbb{R}^n)\times\mathcal{S}^2(\mathbb{R}^n)\times\mathcal{H}^2(\mathbb{R}^{n\times d})$. Moreover, it holds that
 \begin{equation*}
\begin{cases}
\displaystyle &\sum_{i=1}^{n}Q^{q^i}_t=\sum_{i=1}^{n}Q^i_0+\int_0^t\sum_{i=1}^{n}q^i_sds, \\
\displaystyle &\sum_{i=1}^{n}q^i_t=-\frac{\beta}{b}\sum_{i=1}^{n}Q^{q^i}_T+\int_t^T\frac{1}{b}\left(-\lambda\sigma^2_s\sum_{i=1}^{n}Q^{q^i}_s+\frac{n\mu_s}{2}-\sum_{i=1}^{n}\frac{aq^i_t}{2}+\frac{an\tilde{q}_s}{2}\right)ds+\int_t^T\sum_{i=1}^{n}Z^i_sdW_s
\end{cases}
\end{equation*}
Therefore, we have
 \begin{equation*}
\begin{cases}
\displaystyle &\sum_{i=1}^{n}Q^{q^i}_t-\tilde{Q}_t=\int_0^t\left(\sum_{i=1}^{n}q^i_s-\tilde{q}_s\right)ds, \\
\displaystyle &\sum_{i=1}^{n}q^i_t-\tilde{q}_s=-\frac{\beta}{b}\left(\sum_{i=1}^{n}Q^{q^i}_T-\tilde{Q}_T\right)+\int_t^T\frac{1}{b}\left(-\lambda\sigma^2_s\left(\sum_{i=1}^{n}Q^{q^i}_s-\tilde{Q}_s\right)-\frac{a}{2}\left(\sum_{i=1}^{n}q^i_s-\tilde{q}_s\right)\right)ds\\
&\quad\quad\quad+\int_t^T\left(\sum_{i=1}^{n}Z^i_s-\tilde{Z}_s\right)dW_s
\end{cases}
\end{equation*}
It follows from the uniqueness part of \emph{Step 1} that
\begin{equation*}
\begin{cases}
&\tilde{Q}=\sum_{i=1}^{n}Q^{q^i},\\
&\tilde{q}=\sum_{i=1}^{n}q^i,\\
&\tilde{Z}=\sum_{i=1}^{n}Z^i.
\end{cases}
\end{equation*}
\emph{Step 3:} The last statement follows immediately from the uniqueness of solutions of FBSDEs \eqref{FBSDE1} and \eqref{FBSDE2}.
\end{proof}
\subsubsection{Asymptotic property}
If we scale the permanent market impact by the number of agents $n$ or equivalently the permanent market impact is generated by the average of liquidation strategy of all agents, the FBSDE characterizing the Nash equilibrium turns to be the following FBSDE
\begin{equation}\label{FBSDE3}
\begin{cases}
\displaystyle &Q^{q^{i,n}}_t=Q^i_0+\int_0^tq^{i,n}_sds,~~i=1,\ldots,n\\
\displaystyle &q^{i,n}_t=-\frac{\alpha_i-\frac{a}{2n}}{b}Q^{q^{i,n}}_T+\int_t^T\frac{1}{b}\left(-\lambda_i\sigma^2_sQ^{q^{i,n}}_s+\frac{\left(\mu_s+\frac{a}{n}\sum_{j\neq i}q^{j,n}_s\right)}{2}\right)ds+\int_t^TZ^{i,n}_sdW_s,~~i=1,\ldots,n.
\end{cases}
\end{equation}
Then we have the following theorem.
\begin{theorem}
 Suppose that $\alpha_i=\alpha > 0$, $\lambda_i=\lambda\geq 0$ for all $i\in\mathbb{N}$ and $\lim_{n\rightarrow\infty}\frac{1}{n}\sum_{i=1}^nQ^i_0=Q^*_0\in\mathbb{R}$. Then
 \begin{equation}\label{FBSDE4}
\begin{cases}
\displaystyle &Q^{*}_t=Q^*_0+\int_0^tq^*_sds,\\
\displaystyle &q^*_t=-\frac{\alpha}{b}Q^{*}_T+\int_t^T\frac{1}{b}\left(-\lambda\sigma^2_sQ^{*}_s+\frac{\left(\mu_s+aq^*_s\right)}{2}\right)ds+\int_t^TZ^*_sdW_s.
\end{cases}
\end{equation}
admits a unique solution $(Q^*,q^*,Z^*)\in\mathcal{S}^{2}(\mathbb{R})\times\mathcal{S}^{2}(\mathbb{R})\times\mathcal{H}^2(\mathbb{R}^d)$ and
\begin{equation}\label{FBSDE5}
\begin{cases}
\displaystyle &Q^{\tilde{q}^i}_t=Q^i_0+\int_0^t\tilde{q}^i_sds\\
\displaystyle &\tilde{q}^i_t=-\frac{\alpha}{b}Q^{\tilde{q}^i}_T+\int_t^T\frac{1}{b}\left(-\lambda\sigma^2_sQ^{\tilde{q}^i}_s+\frac{\left(\mu_s+aq^*_s\right)}{2}\right)ds+\int_t^T\tilde{Z}^i_sdW_s
\end{cases}
\end{equation}
admits a unique solution $(Q^{\tilde{q}^i},\tilde{q}^i,\tilde{Z}^i)\in\mathcal{S}^{2}(\mathbb{R})\times\mathcal{S}^{2}(\mathbb{R})\times\mathcal{H}^2(\mathbb{R}^d)$ for all $i\in\mathbb{N}$. Let $n$ be large enough such that $\alpha\geq \frac{a}{2n}$ and $(Q^{q^{\cdot,n}},q^{\cdot,n},Z^{\cdot,n})\in\mathcal{S}^{2}(\mathbb{R}^n)\times\mathcal{S}^{2}(\mathbb{R}^n)\times\mathcal{H}^2(\mathbb{R}^{n\times d})$ be the unique solution of FBSDE \eqref{FBSDE3}. Then it holds that
\begin{equation*}
\|\frac{1}{n}\sum_{i=1}^{n}Q^{q^{i,n}}-Q^*\|_{\mathcal{S}^2(\mathbb{R})}+\|\frac{1}{n}\sum_{i=1}^{n}q^{i,n}-q^*\|_{\mathcal{S}^2(\mathbb{R})}+\|\frac{1}{n}\sum_{i=1}^{n}Z^{i,n}-Z^*\|_{\mathcal{H}^2(\mathbb{R}^d)}\rightarrow 0 \text{ as } n\rightarrow \infty
\end{equation*}
and
\begin{equation*}
\|Q^{q^{i,n}}-Q^{\tilde{q}^i}\|_{\mathcal{S}^2(\mathbb{R})}+\|q^{i,n}-\tilde{q}^i\|_{\mathcal{S}^2(\mathbb{R})}+\|Z^{i,n}-\tilde{Z}^i\|_{\mathcal{H}^2(\mathbb{R}^d)}\rightarrow 0 \text{ for all } 1\leq i\leq n, \text{ as } n\rightarrow \infty
\end{equation*}
\end{theorem}
\begin{proof}
It follows from a similar technique as in Theorem \ref{SIM},  FBSDE \eqref{FBSDE4}
admits a unique solution $(Q^*,q^*,Z^*)\in\mathcal{S}^{2}(\mathbb{R})\times\mathcal{S}^{2}(\mathbb{R})\times\mathcal{H}^2(\mathbb{R}^d)$ and FBSDE \eqref{FBSDE5} admits a unique solution $(Q^{\tilde{q}^i},\tilde{q}^i,\tilde{Z}^i)\in\mathcal{S}^{2}(\mathbb{R})\times\mathcal{S}^{2}(\mathbb{R})\times\mathcal{H}^2(\mathbb{R}^d)$ for all $i\in\mathbb{N}$.

 Let $n$ be large enough such that $\alpha\geq \frac{a}{2n}$, it follows from Theorem \ref{SIM} that FBSDE \eqref{FBSDE3} admits a unique solution $(Q^{q^{\cdot,n}},q^{\cdot,n},Z^{\cdot,n})\in\mathcal{S}^{2}(\mathbb{R}^n)\times\mathcal{S}^{2}(\mathbb{R}^n)\times\mathcal{H}^2(\mathbb{R}^{n\times d})$. Moreover, one could check that there exists a constant $M$ which does not depend on $n$ such that
 \begin{equation*}
\|Q^{q^{i,n}}\|_{\mathcal{S}^2(\mathbb{R})}+\|q^{i,n}\|_{\mathcal{S}^2(\mathbb{R})}+\|Z^{i,n}\|_{\mathcal{H}^2(\mathbb{R}^d)}\leq M,\text{ for all } 1\leq i\leq n.
\end{equation*}
In addition, we have
\begin{equation*}
\begin{cases}
\displaystyle &\frac{1}{n}\sum_{i=1}^{n}Q^{q^{i,n}}_t=\frac{1}{n}\sum_{i=1}^{n}Q^i_0+\int_0^t\frac{1}{n}\sum_{i=1}^{n}q^{i,n}_sds,\\
\displaystyle &\frac{1}{n}\sum_{i=1}^{n}q^{i,n}_t=-\frac{\alpha-\frac{a}{2n}}{b}\frac{1}{n}\sum_{i=1}^{n}Q^{q^{i,n}}_T+\int_t^T\frac{1}{b}\left(-\lambda\sigma^2_s\frac{1}{n}\sum_{i=1}^{n}Q^{q^{i,n}}_s+\frac{\mu_s}{2}+\frac{a(n-1)}{2n}\frac{1}{n}\sum_{i=1}^{n}q^{i,n}_s\right)ds\\
&\quad\quad\quad+\int_t^T\frac{1}{n}\sum_{i=1}^{n}Z^{i,n}_sdW_s.
\end{cases}
\end{equation*}
Therefore, it holds that
\begin{equation*}
\begin{cases}
\displaystyle &\frac{1}{n}\sum_{i=1}^{n}Q^{q^{i,n}}_t-Q^*_t=\frac{1}{n}\sum_{i=1}^{n}Q^i_0-Q^*_0+\int_0^t\left(\frac{1}{n}\sum_{i=1}^{n}q^{i,n}_s-q^*_s\right)ds,\\
\displaystyle &\frac{1}{n}\sum_{i=1}^{n}q^{i,n}_t-q^*_t=-\frac{\alpha-\frac{a}{2n}}{b}\frac{1}{n}\sum_{i=1}^{n}Q^{q^{i,n}}_T+\frac{\alpha}{b}Q^*_T-\int_t^T\frac{\lambda\sigma^2_s}{b}\left(\frac{1}{n}\sum_{i=1}^{n}Q^{q^{i,n}}_s-Q^*_s\right)ds\\
&\quad\quad\quad+\int_t^T\frac{1}{b}\left(\frac{a}{2}\left(\frac{1}{n}\sum_{i=1}^{n}q^{i,n}_s-q^*_s\right)-\frac{a}{2n^2}\sum_{i=1}^{n}q^{i,n}_s\right)ds+\int_t^T\left(\frac{1}{n}\sum_{i=1}^{n}Z^{i,n}_s-Z^*_s\right)dW_s.
\end{cases}
\end{equation*}
Thus, we get
\begin{align*}
&\left(\frac{1}{n}\sum_{i=1}^{n}q^{i,n}_t-q^*_t\right)\left(\frac{1}{n}\sum_{i=1}^{n}Q^{q^{i,n}}_t-Q^*_t\right)\\
&=\left(-\frac{\alpha-\frac{a}{2n}}{b}\frac{1}{n}\sum_{i=1}^{n}Q^{q^{i,n}}_T+\frac{\alpha}{b}Q^*_T\right)\left(\frac{1}{n}\sum_{i=1}^{n}Q^{q^{i,n}}_T-Q^*_T\right)-\int_t^T\frac{\lambda\sigma^2_s}{b}\left(\frac{1}{n}\sum_{i=1}^{n}Q^{q^{i,n}}_s-Q^*_s\right)^2ds\\
&\quad+\int_t^T\frac{1}{b}\left(\frac{a}{2}\left(\frac{1}{n}\sum_{i=1}^{n}q^{i,n}_s-q^*_s\right)-\frac{a}{2n^2}\sum_{i=1}^{n}q^{i,n}_s\right)\left(\frac{1}{n}\sum_{i=1}^{n}Q^{q^{i,n}}_s-Q^*_s\right)ds\\
&\quad-\int_t^T\left(\frac{1}{n}\sum_{i=1}^{n}q^{i,n}_s-q^*_s\right)^2ds+\int_t^T\left(\frac{1}{n}\sum_{i=1}^{n}Q^{q^{i,n}}_s-Q^*_s\right)\left(\frac{1}{n}\sum_{i=1}^{n}Z^{i,n}_s-Z^*_s\right)dW_s.
\end{align*}
Therefore, we obtain
\begin{align*}
&e^{\frac{at}{2b}}\left(\frac{1}{n}\sum_{i=1}^{n}q^{i,n}_t-q^*_t\right)\left(\frac{1}{n}\sum_{i=1}^{n}Q^{q^{i,n}}_t-Q^*_t\right)\\
&=e^{\frac{aT}{2b}}\left(-\frac{\alpha-\frac{a}{2n}}{b}\frac{1}{n}\sum_{i=1}^{n}Q^{q^{i,n}}_T+\frac{\alpha}{b}Q^*_T\right)\left(\frac{1}{n}\sum_{i=1}^{n}Q^{q^{i,n}}_T-Q^*_T\right)-\int_t^Te^{\frac{as}{2b}}\frac{\lambda\sigma^2_s}{b}\left(\frac{1}{n}\sum_{i=1}^{n}Q^{q^{i,n}}_s-Q^*_s\right)^2ds\\
&\quad-\int_t^Te^{\frac{as}{2b}}\left(\frac{1}{n}\sum_{i=1}^{n}q^{i,n}_s-q^*_s\right)^2ds-\int_t^Te^{\frac{as}{2b}}\frac{a}{2bn^2}\sum_{i=1}^{n}q^{i,n}_s\left(\frac{1}{n}\sum_{i=1}^{n}Q^{q^{i,n}}_s-Q^*_s\right)ds\\
&\quad+\int_t^Te^{\frac{as}{2b}}\left(\frac{1}{n}\sum_{i=1}^{n}Q^{q^{i,n}}_s-Q^*_s\right)\left(\frac{1}{n}\sum_{i=1}^{n}Z^{i,n}_s-Z^*_s\right)dW_s.
\end{align*}
Hence, it holds that
\begin{align*}
0&\geq E\left[-e^{\frac{aT}{2b}}\frac{a}{b}\left(\frac{1}{n}\sum_{i=1}^{n}Q^{q^{i,n}}_T-Q^*_T\right)^2-\int_0^Te^{\frac{as}{2b}}\frac{\lambda\sigma^2_s}{b}\left(\frac{1}{n}\sum_{i=1}^{n}Q^{q^{i,n}}_s-Q^*_s\right)^2ds-\int_t^Te^{\frac{as}{2b}}\left(\frac{1}{n}\sum_{i=1}^{n}q^{i,n}_s-q^*_s\right)^2ds\right]\\
&=E\left[\left(\frac{1}{n}\sum_{i=1}^{n}q^{i,n}_0-q^*_0\right)\left(\frac{1}{n}\sum_{i=1}^{n}Q^{q^{i,n}}_0-Q^*_0\right)\right]+E\left[e^{\frac{aT}{2b}}\frac{a}{2bn^2}\sum_{i=1}^{n}Q^{q^{i,n}}_T\left(\frac{1}{n}\sum_{i=1}^{n}Q^{q^{i,n}}_T-Q^*_T\right)\right]\\
&\quad-E\left[\int_0^Te^{\frac{as}{2b}}\frac{a}{2bn^2}\sum_{i=1}^{n}q^{i,n}_s\left(\frac{1}{n}\sum_{i=1}^{n}Q^{q^{i,n}}_s-Q^*_s\right)ds\right]
\end{align*}
which goes to $0$ as $n$ goes to infinity. Therefore, one could deduce that
\begin{equation*}
\|\frac{1}{n}\sum_{i=1}^{n}Q^{q^{i,n}}-Q^*\|_{\mathcal{S}^2(\mathbb{R})}+\|\frac{1}{n}\sum_{i=1}^{n}q^{i,n}-q^*\|_{\mathcal{S}^2(\mathbb{R})}+\|\frac{1}{n}\sum_{i=1}^{n}Z^{i,n}-Z^*\|_{\mathcal{H}^2(\mathbb{R}^d)}\rightarrow 0 \text{ as } n\rightarrow 0
\end{equation*}
Similarly, it holds that
\begin{equation*}
\|Q^{q^{i,n}}-Q^{\tilde{q}_i}\|_{\mathcal{S}^2(\mathbb{R})}+\|q^{i,n}-\tilde{q}_i\|_{\mathcal{S}^2(\mathbb{R})}+\|Z^{i,n}-\tilde{Z}_i\|_{\mathcal{H}^2(\mathbb{R}^d)}\rightarrow 0 \text{ for all } 1\leq i\leq n, \text{ as } n\rightarrow 0
\end{equation*}
\end{proof}

\bibliographystyle{abbrv}
\bibliography{MarketImpact}

 \end{document}